\documentclass[a4paper,twocolumn,11pt,accepted=2026-06-04]{quantumarticle}
\pdfoutput=1

\usepackage{amsthm}
\usepackage{amsmath}
\usepackage{amssymb}
\usepackage{amsfonts}

\usepackage{lineno}
\usepackage{array}
\usepackage{graphicx} 
\usepackage{geometry}

\usepackage{hyperref}
\hypersetup{
    colorlinks=true,
    citecolor=blue,
    }

\usepackage{float}

\usepackage{mathtools}
\usepackage{stmaryrd}

\usepackage{tikz}
\usepackage{quantikz}
\usetikzlibrary{arrows.meta,positioning}

\usepackage{appendix}

\usepackage{algorithm}
\usepackage{algorithmic}

\usepackage{booktabs}

\usepackage{braket}

\usepackage{fancyhdr}
\usepackage{todonotes}
\usepackage[capitalize,nameinlink]{cleveref}

\usepackage{multicol}
\usepackage{dutchcal}
\usepackage{xcolor}
\usepackage{tcolorbox}
\tcbuselibrary{skins,breakable}

\newtheorem{proposition}{Proposition}
\newtheorem{lemma}{Lemma}

\newtheorem{theorem}{Theorem}
\newtheorem{definition}{Definition}
\newtheorem{example}{Example}
\newtheorem{corollary}{Corollary}
\newtheorem{remark}{Remark}

\definecolor{reviewercolor}{RGB}{150,90,70}    
\definecolor{revisioncolor}{RGB}{233,150,122}    
\definecolor{light}{RGB}{252,248,246}

\newcommand{\revision}[2][]{%
\begin{tcolorbox}[
  breakable,
  colback=light,
  colframe=revisioncolor,
  title={%
    \textbf{Revision / Response}%
    \ifx#1\empty\else\hfill\checkmark\fi
  },
  fonttitle=\bfseries,
  left=1mm,
  right=1mm,
  top=1mm,
  bottom=1mm
]
#2
\end{tcolorbox}\vspace{0.5cm}
}

\makeatletter
\newcommand{\settitle}{\@maketitle}
\makeatother

\newcommand{\Span}{\text{span}}

\newcommand{\Code}{\mathcal{C}}
\newcommand{\CNOT}{\mathsf{CNOT}}
\newcommand{\CZ}{\mathsf{CZ}}
\newcommand{\CU}{\mathsf{C}U}
\newcommand{\CSS}{\mathsf{CSS}}
\newcommand{\SWAP}{\mathsf{SWAP}}
\newcommand{\Sym}{\text{Sym}}

\newcommand{\Stab}{\mathcal{S}}

\newcommand{\F}{\mathbb{F}}

\title{On the addressability problem on CSS codes}

\author{Jérôme Guyot}
\email{jerome.guyot@ens-paris-saclay.fr}
\affiliation{ENS Paris-Saclay, Université Paris-Saclay,  France }

\author{Samuel Jaques}
\email{sejaques@uwaterloo.ca}
\orcid{0000-0003-0966-8114}
\affiliation{University of Waterloo, Canada}

\date{}

\begin{document}

\maketitle

\begin{abstract}

Recent discoveries in asymptotically good quantum codes have intensified research on their application in quantum computation and fault-tolerant operations. This study focuses on the addressability problem within CSS codes: we ask what circuits might implement logical gates on strict subsets of logical qubits. With some notion of fault-tolerance, we prove several impossibility results: for CSS codes with non-zero rate, one cannot address a logical $H$, $HS$, $SH$, or $\CNOT$ to any non-empty strict subset of logical qubits using a circuit made only from 1-local Clifford gates. 

Furthermore, we show that one cannot permute the logical qubits in a code purely by permuting the physical qubits, if the rate of the code is (asymptotically) greater than $\frac{1}{3}$ and the distance is at least 3. We can show a similar no-go result for CNOTs and CZs between two such high-rate codes, albeit under a more restrictive assumption on the circuit, which we call ``global" (though recent addressable CCZ gates use global circuits).

This work pioneers the study of distance-preserving addressability in quantum codes, mainly by considering automorphisms of the code. This perspective offers new insights and potential directions for future research. We argue that studying this trade off between addressability and efficiency of the codes is essential to understand better how to do efficient quantum computation.

\end{abstract}


\section*{Introduction}

\subsection*{Motivation}
Quantum computers are particularly vulnerable to noise, and so the most promising path to large-scale quantum computing is to use error-correcting codes. In these codes, many \emph{physical} qubits are combined into one or more \emph{logical} qubit(s), such that the logical qubits are long-lived and error-resistant.

A drawback of quantum error correction is that, by design, it becomes difficult to modify the encoded logical state. Unlike with classical error-correcting codes, we cannot decode the state to compute on it, as it is unlikely to remain coherent long enough for any operation. Thus, we need fault-tolerant quantum computation: not only should we have a method to encode the data, but we should also be able to operate on it \emph{while} it is encoded. 

A powerful tool in constructing fault-tolerant quantum computation is a \emph{transversal} gate. Strictly speaking, this is any physical circuit that is guaranteed not to propagate errors between qubits of the code, and more commonly we require that it enacts some specific action on the logical state as well. 

As an example, in a self-dual CSS code, applying an $H$ (Hadamard) gate to all physical qubits in the code will not only preserve the codespace, it will effectively apply an $H$ gate to all \emph{logical} qubits in the code. However, in most quantum circuits we need more precision than this. We need to be able to apply specific gates to \emph{only one} qubit in the code. We need our fault-tolerant operations to be \emph{addressable}.

This distinction does not matter for surface codes, which (depending on the precise description) encode only one logical qubit. A large-scale surface code computation is best seen as a collection of codes, each working independently. We will later use the words \emph{splitting codes} when, as in this example, the global code can be ``cut" into independent sub-codes. In particular, for this collection of independent surface codes, any transversal gate can be targeted to a single logical qubit (or pair of qubits for a CNOT) by simply applying the gates only to those physical qubits corresponding to the desired logical qubit. 

However, this strategy fails for more complicated codes that encode many qubits, as each physical qubit no longer corresponds to only one logical qubit, as different logical operators will overlap in their supports (e.g., the overlaps in \Cref{fig:image}). For good performance, the logical qubits are not spatially localized in this way. If we apply a gate to one physical qubit, it will impact many logical qubits.

Addressability is crucial for efficient quantum computation, as we need the flexibility to apply any gate to any logical qubit. Transversality only guarantees efficiency, but is limited in its expressive power, as the Eastin-Knill theorem highlights. Addressability provides a more detailed view of a code’s structure by considering each logical qubit individually. Studying addressability helps identify fundamental trade-offs in designing fault-tolerant operations for high-rate quantum codes, ensuring that logical operations remain precise and scalable in larger quantum systems.

\begin{figure}[H]
\centering
        \includegraphics[width=\columnwidth]{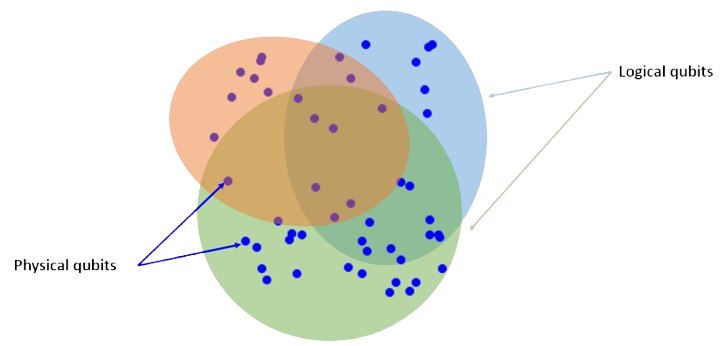}
        \caption{Visualization of a code}
        \label{fig:image}
\end{figure}

This addressability problem has become relevant recently, with the development of asymptotically good quantum codes~\cite{quantum_tanner_codes,asymptotically_good_ldpc_PK}. As the size of the code increases, the number of logical qubits in these codes approaches the number of physical qubits up to a constant factor (strictly smaller than 1), while maintaining good code distance, which means logical qubits cannot be spatially localized. Thus, it is a challenging problem to find efficient physical circuits which can address specific logical qubits for some desired logical gate, and that is what we address in this paper.

We highlight that the addressability problem is only interesting when restricting the set of possible implementations. In particular, for any code, the trivial ``decode, then apply gates, then encode'' method could work. However, this method is neither efficient nor fault-tolerant. Similarly, addressability is often reached using code-surgeries or magic state injection methods, which are outside of the scope of this study. Here we consider fault-tolerant unitary implementations.

\subsection*{Methods}
We consider the problem of efficient physical circuits to enact logical gates addressably. There are two trivial ways to make addressable gates: the first is to decode, apply the gate, then re-encode, and the second is to use a collection of independent codes which admit transversal gates. 

The first trivial method is problematic because it is not fault tolerant. The quantum state is unprotected after decoding. Thus, we consider gate sets that will not alter the distance of code. We take three approaches for this: first, 1-local circuits (i.e., circuits built from single-qubit gates); second, circuits made of SWAPs or other permutations; third, circuits made from a depth-1 CNOT or CZ circuit.

The second trivial method, using a collection of independent codes, forbids us from having high-rate codes. One can readily see that if multiple codes are run in parallel, and we treat them as one larger code, the larger code's distance is at most the minimum distance from one of the subcodes. Thus, asymptotically good codes cannot ``split'' into subcodes like this\footnote{Or at least, they must contain only asymptotically good sub-codes which do not split.}.

We also notice that if we take any CSS code and apply any circuit of single-qubit gates, we will obtain another code with the same distance and rate. However, this new code may not be easy to work with and it may not have any of its own efficient fault-tolerant operations. Thus, we further require that our circuits preserve the original code space. 

Ultimately, we are considering operators that preserve the code space, for CSS codes that do not split. 

\subsection*{Results}
We show a series of impossibility results for addressable gates under the given restrictions.

We start with $1$-local Clifford circuits: circuits made only from single-qubit Clifford gates. Since these have well-known commutation relations with Paulis, our main technique is to apply those relations to stabilizers and logical operators and ask when the output is consistent with the codespace and the desired logical action. We conclude that for any non-splitting code:
\begin{itemize}
    \item Applying $1$-local Clifford circuits which preserve the code cannot apply $H$, $SH$, or $HS$ to a strict subset of physical qubits (\Cref{prop:no-partial-from-H-type}).
    \item  No $1$-local Clifford circuit can enact an addressable logical $H$, $SH$, $HS$, or $\CNOT$ gate (\Cref{prop:no-h-cnot}).
    \item If the code does not admit any logical identities made of 1-local circuits using only $S$ or $SHS$ gates, then any 1-local circuit made of gates from the third level of the Clifford hierarchy and beyond does not preserve the codespace (\Cref{cor:P-or-nothing-Clifford-hierarchy}).
\end{itemize}

As introduced in \cite{markus}, automorphisms of the code can be used to implement logical operations using permutations. Our second set of results uses this principle, just as in \cite{markus} or \cite{photonic2025swapCNOT} which uses permutations to construct addressable Clifford gates. We show that the number of permutations which can preserve the codespace \emph{and} produce distinct logical actions is, asymptotically, quite limited: for codes with rates above $\frac{1}{3}$ and distance at least 3, it is less than $k!$ when the code has $k$ logical qubits.

With this idea, we can make several conclusions about CSS codes with distance at least 3:
\begin{itemize}
    \item Circuits made from only SWAP gates (or any permutations of physical qubits) cannot implement all logical permutations on a code with rate asymptotically greater than $1/3$. They also cannot implement $\CNOT(i,j)$ for all pairs of qubits $(i,j)$ on codes with any constant rate. Finally, for any $2$-qubit gate $G$ such families of circuits cannot implement $G(i,j)$ for all pairs of qubits $(i,j)$ on codes with rates greater than $3/4$.
    \item Circuits between two codes, made by applying CNOT (resp. CZ) to all physical qubits, cannot implement all depth-$1$ logical CNOT (resp. CZ) circuits on codes with rate asymptotically greater than $\frac{1}{3}$ (\Cref{prop:no-global-cnot,prop:no-global-CZ}).
\end{itemize}
While the second result seems more restrictive, since the circuit must act on all physical qubits, the circuits for transversal CCZ gates from \cite{hvwz2025ccz} satisfy this property, as would a CZ built from their methods. However, they do not aim for parallel addressability.

While a rate of $\frac{1}{3}$ (and especially $\frac{3}{4}$) is too high to restrict most codes, this bound comes from a rather loose estimate on the bounds of certain code automorphisms (\Cref{upp nbr permutation}). Improved counting arguments would immediately give a stronger restriction on the code's rate.

\subsection*{Related Work}
While considerable work has been done to find valid transversal implementations and study their actions on the code, very little has been done on the addressability problem until recently. Addressability is mentioned when a transversal gate allows it, such as in \cite{Zhu2023NonCliffordAP,markus}, but is rarely a goal itself. Previous works finding addressable gates like \cite{patra2024targetedcliffordlogicalgates,Quintavalle2023partitioningqubits} were not fault-tolerant. Closely related to our study, \cite{synthesis_clifford} shows that on stabilizer codes, any logical Clifford can be implemented using a physical Clifford circuit, thus achieving addressability for logical Cliffords.
Both \cite{lin2024transversalnoncliffordgatesquantum} and \cite{hvwz2025ccz} produce codes with addressable CCZ gates. While \cite{lin2024transversalnoncliffordgatesquantum} is able to address disjoint triples of logical qubits, \cite{hvwz2025ccz} is able to address any triple of logical qubits. In \cite{hvwz2025ccz} they seek many of the same goals as we do, but with constructive results: their gates are fault tolerant, their codes can achieve constant rates, and they can address CCZ gates to arbitrary triples of logical qubits. Though, we considered an even more stringent requirement: can one address arbitrary disjoint subsets of logical qubits \emph{simultaneously}?

The pre-eminent impossibility result for computing on encoded data is the Eastin-Knill theorem~\cite{eastin-knill}, forbidding a set of transversal gates from being universal. We pursue even stronger impossibility results, similar to \cite{PRL:BraKon2013}, who prove that constant-depth operations on topological codes in $D$ dimensions can only implement gates from up to the $D$th level of the Clifford hierarchy. As it turns out, high-rate codes are necessarily non-local \cite{PRL:BasKri2022}, hence \cite{PRL:BraKon2013} will not provide the techniques needed to tackle these codes. 

\cite{PRX:OCoKubYod2018} gives bounds in terms of the ``disjointness'' of the logical Pauli operators, and also prove that any transversal gate on a stabilizer code must be in the Clifford hierarchy. Combined with our \Cref{cor:P-or-nothing-Clifford-hierarchy}, the $S$ gate acts like a load-bearing gate: if it is not addressable, little else can be. 

One major departure in our approach from the impossibility results in \cite{PRL:BasKri2022,PRX:OCoKubYod2018} is that these papers consider arbitrary physical implementations (though perhaps limited by depth), whereas we focus on particular physical implementations of the addressable gates. This makes our results less general, but they can be more powerful when they do apply. Focusing on specific physical gates gives us much more control over their effects on the stabilizers. So far this perspective seems limited to \emph{constructive} results like \cite{markus,hvwz2025ccz,lin2024transversalnoncliffordgatesquantum,patra2024targetedcliffordlogicalgates,Quintavalle2023partitioningqubits,Zhu2023NonCliffordAP}; we expect future ``pessimistic'' research from this perspective could yield more impossibility results.

A large amount of (nearly) concurrent work has appeared on this problem recently. Closest to our results is a no-go theorem that $k-1$-fold transversal Clifford circuits cannot implement the full logical Clifford group on a code with $k$ or more logical qubits~\cite{ARXIV:ChaGot26}. They use entirely different techniques and the result is more general, but also less specific: we prove no-go theorems for specific logical Cliffords. A complementary result constructs a quantum Reed-Muller where fold-transversal gates \emph{can} implement the full logical Clifford group, by using a more permissive notion of ``fold-transversal''~\cite{ARXIV:TanChaTak26}. One of our restrictions was requiring that the physical circuit preserves the code perfectly. \cite{ARXIV:TanChaTak26} work allows physical Clifford gates to modify the code, but then they restore the original code with a qubit permutation.

Separately, a large number of constructive results have appeared recently for performing logical operations on high-rate codes, using higher-depth circuits~\cite{photonic2025swapCNOT,PRXQ:SwaJocYog26}, teleporting gates from other codes~\cite{ARXIV:PGPP25,PRX:XZZB+25}, code switching~\cite{ARXIV:LiPreXu25,ARXIV:THLGH25}, lattice surgery~\cite{ARXIV:CHRY25}, or combinations of these techniques~\cite{ARXIV:XZBC+25}. Most of these add some overhead in space or time compared to the small, unitary circuits that we consider; for example, \cite{ARXIV:PGPP25} provide a technique to teleport any Clifford gate onto a specific code using an auxiliary Bacon-Shor code, but this requires extra qubits for the auxiliary code and it is unclear whether multiple logical qubits can be addressed in parallel. To answer whether this is a practical improvement on other codes, one needs to do a full accounting of these overheads (notably, \cite{PRX:XZZB+25} do exactly this analysis for their code by constructing addition circuits).

Our perspective was the reverse: we take several restrictive assumptions on the technique to force it to be efficient (one-local, unitary, Clifford, etc.), then ask whether such an efficient operation can do something useful.

\subsection*{Conclusions}
Ideally, we would answer the question of addressability, by either giving a method to perform addressable gates on high-performance codes, or definitively proving that this is impossible. Instead, we have only \emph{some} impossibility results. However, our results suggest what routes will be necessary if addressable gates are possible, highlight new proof techniques for considering these problems, and emphasize some of the restrictions we might need in considering the addressability problem.

For example, \cite{patra2024targetedcliffordlogicalgates} and \cite{Quintavalle2023partitioningqubits} seem to contradict our results by providing an addressable $H$ gate. However, as these papers point out themselves, their techniques do not necessarily preserve distance. \cite{patra2024targetedcliffordlogicalgates} involves enacting a linear transformation on the stabilizer vectors by applying a physical CNOT from each physical qubit in the code to \emph{an unprotected} auxiliary qubit. This means any phase error on this qubit will propagate up into the code. Hence, distance-preserving techniques remain an important consideration.

One easy fix might be to encode the auxiliary qubit in a different code (say, a surface code). In \cite{ARXIV:PGPP25} they use a Bacon-Shor code to teleport gates into a qLDPC code using targeted CNOTs and measurements. This does not seem to allow \emph{simultaneous} addressability.

Our impossibility results on $1$-local Clifford addressability shows that the algorithms of \cite{synthesis_clifford} or \cite{ARXIV:TanChaTak26} cannot, in general, be refined to  output $1$-local Clifford circuits. Extending our analysis to bounded-depth Clifford circuits would yield a tighter understanding of how optimal their algorithm is.

In concurrent work, \cite{hvwz2025ccz} uses a depth-1 physical CCZ circuit for both ``intra-code'' and ``inter-code'' addressable logical CCZ gates. They prove a constructive result for a constant code rate, whereas our CZ impossibility results apply to codes with higher rates than they construct. Despite the similarities between our impossibility results and their constructions, our results do not apply to their codes. The main difference is that our results forbid what we call ``parallel addressability'' (\Cref{def:parallel-addressability}), where if two logical gates act on disjoint sets of logical qubits, we can apply both simultaneously. Their construction has some ability to do this, but not completely. 

In another concurrent work, \cite{photonic2025swapCNOT} constructs codes with fault-tolerant circuits for addressable Clifford gates constructed from permutation automorphisms. Our results (\Cref{coro style}) show that this method can only work for codes with asymptotically low rates (in $o(1)$), and our upper bound is not far from the rate of their codes.

One method to escape our restrictions would be to allow the physical circuit to modify the code. For example, maybe there is a family of codes that can all be reached from each other by depth-1 CNOT circuits. Indeed, \cite{ARXIV:TanChaTak26} find quantum Reed-Muller codes that allow this. In general, this is a special structure, so we assumed it did not exist, but as \cite{ARXIV:TanChaTak26} shows this could be a productive avenue for constructure results.

Overall, we hope our results motivate more consideration of addressability and that our techniques can be taken further, either for constructive results or impossibility theorems. This work also shows that we should not take it for granted that, because one quantum error-correcting code is able to encode logical qubits more efficiently than another one, it will be overall more efficient for computation.

\section{Background}
\subsection{Notation}

Let $\mathcal{P}_n$ be the set of $n$-qubit Pauli operators. 

We define the Clifford hierarchy inductively as follows: $\mathcal{C}_n^1=\mathcal{P}_n$, and for $k>1$, 
\begin{equation}
    \mathcal{C}_n^k = \{U\in\mathcal{U}_{2^n} | U\mathcal{P}_nU^\dagger\subseteq \mathcal{C}_{k-1}^n\}.\label{eq:clifford-hierarchy}
\end{equation}

We call $\mathcal{C}_n^k$ the $k$th \emph{level} of the Clifford hierarchy, and we simply call $\mathcal{C}_n^2$ the Clifford gates. 

A $k$-\emph{local} circuit is a circuit composed of gates such that each gate acts on at most $k$ qubits. We say a circuit is \emph{global} for a set of qubits if it applies a non-identity gate to all qubits in the set.

We let $\llbracket n\rrbracket$ denote the set $\{1,2,\dots, n\}$.

For a vector $a\in \F_2^n$ and a single-qubit gate $G$, we let $G^a$ denote the operator $\otimes_{i:a_i=1}G_i$, where $G_i$ is $G$ applied to qubit $i$.

For a vector $a\in \F_2^n$, and a set $h\subseteq \llbracket n\rrbracket$, we will sometimes use $a\cap h$ to denote a vector in $\F_2^n$ such that $(a\cap h)_i=1$ if and only if $a_i=1$ and $i\in h$.

A quantum code on $n$ physical qubits, encoding $k$ logical qubits, with distance $d$, is denoted as a $\llbracket n,k,d\rrbracket$ code.

\subsection{CSS Codes}\label{sec:CSS}
We will work entirely with CSS codes~\cite{CSS_1}. A CSS code is constructed from two classical codes $\Code_X,\Code_Z\subseteq \F_2^n$ such that $\Code_X^\perp \subseteq \Code_Z$. Let $H_X$ and $H_Z$ be the parity check matrices of $\Code_X$ and $\Code_Z$.

We will let $\mathcal{S}_X = \{X^a | a\in \Code_X^\perp\}$ and $\mathcal{S}_Z = \{Z^b | b\in \Code_Z^\perp \}$, the stabilizers are generated by the parity check matrices $H_X,H_Z$. We define the code $\CSS(\Code_X,\Code_Z)$ to be the set of quantum states in the $+1$ eigenspace of all operators in $\Stab = \Stab_X \cup \Stab_Z$. 

The orthogonality condition implies that all operators in $\Stab_X$ commute with all operators in $\Stab_Z$.

Using generalized Paulis for the stabilizers, CSS codes can also be defined on qudits from codes on $\F_q^n$ \cite{rains1999quditcodes}.

Starting from section 3, we will work exclusively with the stabilizer spaces of CSS codes rather than the classical codes themselves. To simplify, we introduce the notation $\overline{\CSS}(A,B)$ to denote a CSS code whose $X$-stabilizers are generated by elements of $A$ and $Z$-stabilizers by elements of $B$. Concretely, given a conventional $\CSS(\Code_X, \Code_Z)$ code, we set $A = \Code_X^\perp$ and $B = \Code_Z^\perp$, so that the orthogonality condition $\Code_X^\perp \subseteq \Code_Z$ becomes simply $A \subseteq B^\perp$. This shift in perspective, from classical codes to their duals, is natural given that our constructions and proofs operate directly on stabilizer spaces, and avoids the notational overhead of carrying perpendicular superscripts throughout.

We define a logical operator to be any operator which preserves the codespace. If an operator acts as the identity on the codespace, we call it
a logical identity. In particular, logical operators are the set of operators preserving logical identities by conjugation.

\begin{proposition}\label{cond_valid_unitary}
    $L$ is a logical operator for a code $\Code$ if and only if $LI(\Code)L^\dagger \subseteq I(\Code)$ where $I(\Code)$ is the set of logical identities for the code $\Code$.
\end{proposition}
\begin{proof}
    Let $\ket{\psi}$ be a codeword and $L$ a logical operator. Then $L\ket{\psi}=\ket{\phi}$ for $\ket{\phi}\in \Code$. Let $s\in I(C)$. Then $sL^\dagger\ket{\phi}=L^\dagger\ket{\phi}$, so $LsL^\dagger\ket{\phi}=\ket{\phi}$. Thus, $LsL^\dagger$ is a logical identity.

    Conversely, if $LI(C)L^\dagger\subseteq I(C)$ for an operator $L$, then for any $s\in I(C)$, there is $s'\in I(C)$ such that $sL=Ls'$, so $sL\ket{\psi}=Ls'\ket{\psi}=L\ket{\psi}$ for any $\ket{\psi}\in\Code$. Since the stabilizers for the code are included in $I(C)$, this implies $L\ket{\psi}\in\Code$, so $L$ is a logical operator.

\end{proof}

The \emph{normalizer} of a group $G$ contained in a group $E$ is the set of all $h\in E$ such that $hGh^{-1}\subseteq G$, and is denoted $N_E(G)$ or just $N(G)$ if $E$ is clear from context. 

For a stabilizer code, $N_{\mathcal{P}_n}(\mathcal{S})$ contains all logical Pauli operators on the code. More precisely for CSS codes, $N_{\mathcal{P}_n}(\mathcal{S}_X)$ are all the combinations of $X$ operators and logical Pauli-$Z$ operators, and $N_{\mathcal{P}_n}(\mathcal{S}_Z)$ are all the combinations of $Z$ operators and logical Pauli-$X$ operators. Quotienting by the stabilizers gives distinct logical Pauli operators as cosets of this space. 

From this we can prove standard results about stabilizer codes. For example, each logical Pauli-$Z$ string in $N(\Stab_X)$ corresponds to a vector in $\Code_X$ since the stabilizers in $\Stab_X$ correspond to vectors in $\Code_X^\perp$. Thus, in order to commute with all $X$ stabilizers, the logical Pauli-$Z$ string must be in $\Code_X$. Furthermore, if $\Code_X,\Code_Z$ have dimension $k_X,k_Z$ then $\Code_Z^\perp$ has dimension $n-k_Z$. The logical Pauli-$Z$ operators are the cosets of $N(\Stab_X)/\Stab_Z$, which thus corresponds to $\Code_X/\Code_Z^\perp$ and has dimension $k_X+k_Z-n$. Each generator of $N(\Stab_X)/\Stab_Z$ corresponds to a logical operator and thus a logical qubit, so the code has $k_X+k_Z-n$ logical qubits, a fact we will use throughout this paper.

\subsection{Code Rates}
The \emph{rate} of a code is the number of logical qubits $k$ divided by the number of physical qubits $n$.

\begin{proposition}\label{prop:code-rates}
    Let $\Code=\CSS(\Code_X,\Code_Z)$, and let $\rho',\rho''$ be the maximum and minimum of the rates of the classical codes of $\Code_X$ and $\Code_Z$. Letting $\rho$ be the rate of $\Code$, we have that $\rho=\rho'+\rho''-1$ and $2\rho''-1\leq \rho\leq 2\rho'-1$.
\end{proposition}
\begin{proof}
    The rate of the CSS code is given by is $\rho=\frac{k}{n}=\frac{k_X+k_Z-n}{n}$ where $k_X=\dim(\Code_X)$ and $k_Z=\dim(\Code_Z)$. Using that $\rho_X = \frac{k_X}{n}$ and $\rho_Z = \frac{k_Z}{n}$ we obtain $\rho=\rho_X+\rho_Z-1$. This directly gives $2\rho''-1\leq \rho\leq 2\rho'-1$.

\end{proof}

\subsection{Transversality}
Informally, a ``transversal'' gate is any gate which efficiently implements a logical qubit using physical gates, and typically refers to the case where we apply some physical gate $U$ to all physical qubits and obtain the action of $U$ on the logical qubits. However, transversality has a more general definition (\cite{PRX:OCoKubYod2018})

\begin{definition}[Transversality]
Let $\mathcal{Q}=(Q_i)_{i\in I}$ be a partition of the qubits in a code. We say that a gate $U$ is transversal with respect to $\mathcal{Q}$ if it can be decomposed as $U=\otimes_{i\in I}U_i$ where $U_i$ acts only on $Q_i$.

If $\mathcal{Q}$ is not mentioned explicitly, it is taken to be $Q_i=\{i\}$ for $i=1$ to $n$, or for multiple qubit gates with $p$ blocks of a code, it is taken to be $Q_i=\{i_1,i_2,\dots,i_p\}$ for $i=1$ to $n$.
\end{definition}

\section{Addressability}
We have previously defined a logical operator as an operator which preserves the codespace. We can now ask what action it has. First we recall that any operator $U$ can be written as a linear combination of Paulis, i.e.,
\begin{equation}
    U = \sum_{U_i\in\mathcal{P}_k}\alpha_i U_i
\end{equation}
for coefficients $\alpha_i \in \mathbb{C}$. Thus, we say that a physical circuit $G$ has the logical action of $U$ on a code $\Code$ if $G$ is a logical operator on $\Code$, and for any $\ket{\psi}$ in the code
\begin{equation}
    G\ket{\psi}=\sum_{U_i\in\mathcal{P}_k}\alpha_i\overline{U_i}\ket{\psi}
\end{equation}
where $\overline{U}_i$ is a logical operator for the Pauli $U_i$.

\Cref{fig:targeted-circuit} visualizes such a circuit.

\begin{figure*}[t]
\centering
        \includegraphics[width=0.7\textwidth]{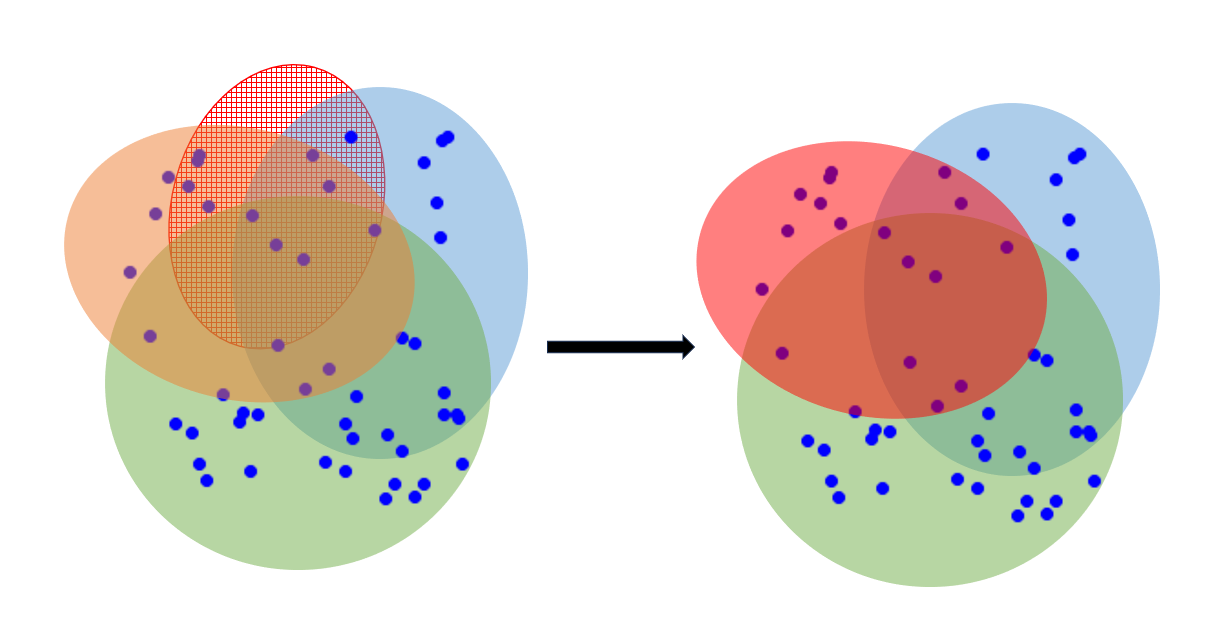}
        \caption{As in \Cref{fig:image} we visualize a code with 3 logical qubits by the solid colored regions (with the coloring denoting the logical state of that qubit), and a physical circuit by the red hatched region. The physical circuit targeted the orange qubit: despite acting on some of the physical qubits in the blue logical qubit, the blue qubit does not change its logical state, while the state of the orange logical qubit changes from orange to red.}
        \label{fig:targeted-circuit}
\end{figure*}

\begin{definition}[Addressability]

    Let $U$ be a unitary on $m$ qubits, and $\mathcal{F}$ be a family of quantum circuits. We say that $U$ is $\mathcal{F}$-\emph{addressable} if, for any ordered tuple of $m$ logical qubits $t$, there is a circuit in $\mathcal{F}$ implementing $\Bar{U}_t$ : the logical action $\Bar{U}$ on the tuple $t$.

\end{definition}

\Cref{fig:addressable} shows what it would mean to be addressable.
\begin{figure*}[t]
\centering
        \includegraphics[width=0.9\textwidth]{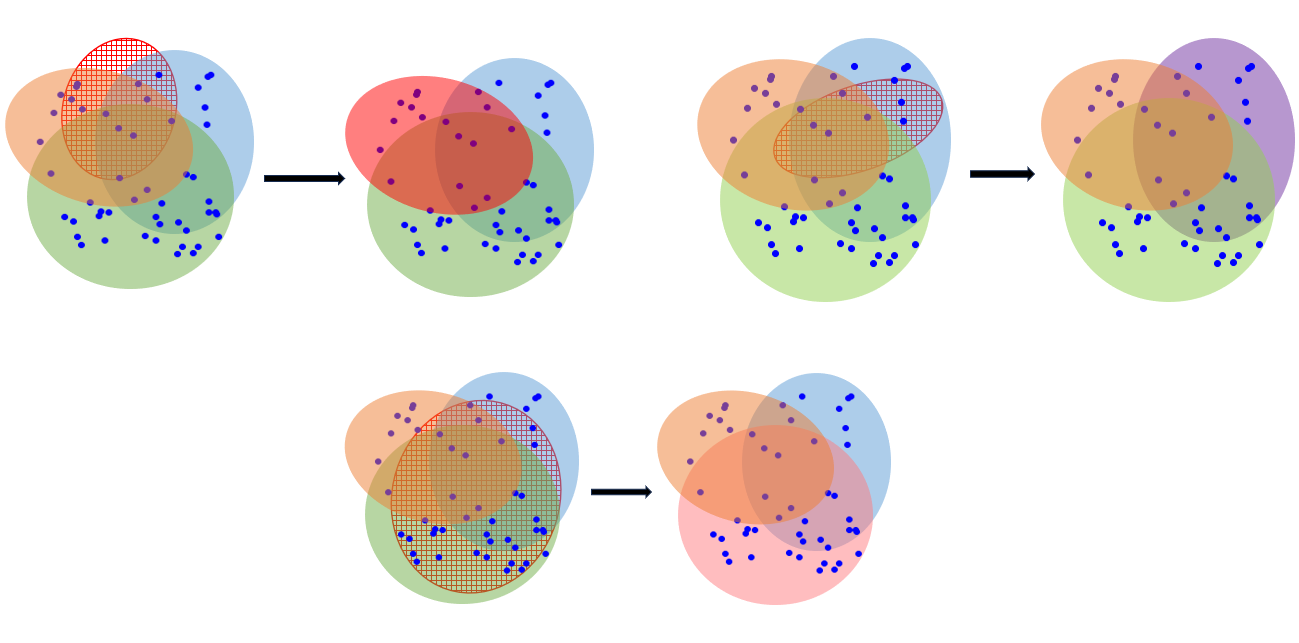}
        \caption{As in \Cref{fig:targeted-circuit} we visualize a code with 3 logical qubits by the solid colored regions (with the coloring denoting the logical state of that qubit), and a physical circuit by the red hatched region. The physical circuits target only one logical qubit each time and act as some unitary that sends 'orange' to 'red', 'blue' to 'purple' and 'green' to 'pink'. In this case, this unitary is addressable on the code.}
        \label{fig:addressable}
\end{figure*}

The restriction to a circuit family $\mathcal{F}$ ensures we do not capture the trivial circuits of decode-apply-encode. Some more useful families of circuits might be:

\begin{enumerate}
    \item Circuits made from Clifford gates (\Cref{sec:cliffords})\label{property:clifford}
    \item Circuits with depth less than $n$\label{property:depth}
    \item Circuits with fewer than $n$ gates\label{property:gates}
    \item Circuits made from SWAPs (\Cref{sec:isomorphisms})\label{property:swaps}
    \item Circuits acting on all physical qubits (\Cref{sec:automorphism-cnots})\label{property:global}
    \item $k$-fold transversal circuits, where the qubits can be partitioned into sets of size at most $k$ where the support of each gate is restricted to one set in the partition. (\cite{ARXIV:ChaGot26})\label{property:kfold}
\end{enumerate}

Some of these properties are not closed under composition (\ref{property:depth},\ref{property:gates}), while others (\ref{property:clifford},\ref{property:swaps},\ref{property:global},\ref{property:kfold}) are closed. To see why this matters, suppose we had an addressable $T$ gate from depth-2 circuits, and we wanted to apply $T$ gates to logical qubits 1 and 2. At the logical level, these gates \emph{should} commute with each other and we should be able to apply them in parallel, but there is no guarantee that the physical circuit will still have depth 2. Thus, we give a stronger definition:

\begin{definition}[Parallel Addressability]
    Let $U$ be a $m$-qubit unitary and $\mathcal{F}$ a family of circuits. We say that $U$ is $\mathcal{F}$-\emph{parallel addressable} on $\Code$ if, for any set $I$ of \emph{disjoint} ordered tuples of $m$ logical qubits, there is a circuit in $\mathcal{F}$ which has the logical action $\Bar{U}$ applied to all tuples in $I$.
\end{definition}

It is easy to see that these definitions are equivalent if $\mathcal{F}$ is closed under composition:
\begin{lemma}
    Let $U$ be a $m$-qubit unitary and $\mathcal{F}$ a family of circuits which is closed under composition. Then $U$ is $\mathcal{F}$-addressable if and only if $U$ is $\mathcal{F}$-parallel addressable.
\end{lemma}
\begin{proof}
    The direction parallel-addressability-to-addressability is implied by definition. Conversely, if $U$ is $\mathcal{F}$-addressable, then for any disjoint set of $m$-tuples of logical qubits $\{I_k\}$, we can find circuits in $\mathcal{F}$ for each tuple. Composing these together will have the required logical action, and the composed circuit will also be in $\mathcal{F}$.
\end{proof}

We also give a restricted definition of addressability to capture cases where we may be able to target \emph{some} subsets of the logical qubits, but not all of them.

\begin{definition}[Partial Addressability]\label{def:parallel-addressability}
    Let $U$ be a $m$-qubit unitary and $\mathcal{F}$ a set of circuits. We say that $U$ is $\mathcal{F}$-\emph{partially addressable} on $\Code$ if there exists a set $I$ of disjoint ordered tuples of $m$ logical qubits ($I\neq \emptyset,\llbracket k\rrbracket$) and a quantum circuit in $\mathcal{F}$ that has the logical action $\Bar{U}$ applied to all tuples in $I$.
\end{definition}

\begin{example}
    \cref{fig:targeted-circuit} shows that $U$ is partially addressable on this code since there is a targeting circuit implementing $U$ on $I= $\{ (Orange logical qubit) \}.   
\end{example}

Finally, we make a brief note about basis. The definition of a logical action requires a particular choice of ``basis'', i.e., which logical operators in $N(\mathcal{S})$ represent which logical Pauli operator. A different choice would change the logical action of an otherwise addressable gate; thus, our definitions are basis-dependent.

However, we do not see this as a limitation: firstly, some of our results forbid \emph{any} logical action from certain circuit families, which is inherently a basis-independent result. Second, if we are allowed to modify the basis of a code, any gate becomes possible. For example, instead of performing a logical $H$ gate, we could re-define the basis of the code to swap those $X$ and $Z$ stabilizers. This would require modifying all future gates, equivalent to commuting the $H$ through all subsequent gates in the circuit. While such compilation can be extremely useful (such as $X$ and $Z$ gates in the surface code), we cannot efficiently compute this for \emph{all} gates in a quantum circuit, or else quantum circuits would be efficiently classical simulatable! Hence we restrict to a single basis.

\subsection{Comparison to Transversality}
One can think of two dimensions of fault-tolerant operations: the utility of the operation and the practicality of the operation.

The original definition of transversality focuses on the practicality: we simply ask that a physical gate can be partitioned in such a way that it preserve's the code distance. On a case-by-case basis we make requirements on its logical action, i.e., its utility for computation.

Addressability focuses on the utility: we require a specific logical action, and we require that action on specific qubits (ideally, all of them independently). We consider its practicality afterwards, by restricting which gates we are allowed to use to make the gate.

\section{Splitting Codes}
We now consider the second main restriction. If a code $\Code$ is simply two codes $\Code_1$ and $\Code_2$ run in parallel, then one of these two codes will have parameters at least as good as $\Code$.  Hence, we want to discard such codes, as our motivation is performing computations on asymptotically good codes. If we are willing to sacrifice code performance for ease of addressable gates, the surface code is a great choice.

One of our main proof techniques is proving that the only way to have partial addressability for certain gates is if the code has this product structure. We say such a code ``splits'' (as shown in \Cref{fig:splitting-code}).

From now on, we use the $\overline{\CSS}(A,B)$ notation introduced in \cref{sec:CSS}, where $A = \Code_X^\perp$ and $B = \Code_Z^\perp$ 
represent the stabilizer spaces of the code. 

\begin{definition}[Splitting]
    Let $\Code=\overline{\CSS}(A,B)$. We say that $A$ splits on some non-empty support $h\subsetneq \{1,\dots, n\}$ if $A$ can be written as $A_1\oplus A_2$, where $h$ is the support of $A_1$.

    If $A$ and $B$ both split on the support $h$, we say that the stabilizer group $\mathcal{S}$ and the code $\Code$ split on $h$. 
\end{definition}

We see that this definition is equivalent to saying that $\Code=\Code_1\otimes\Code_2$, where $\Code_1$ and $\Code_2$ are both CSS codes. Here the tensor product is the tensor product of the quantum states that make up the code, not a tensor product of the codes in a homology sense.

\begin{example}
    Let us consider the following CSS codes where we are given the parity check matrices for the $X$ and $Z$ stabilizers as $A$ and $B$. We first show in \cref{matrix_split} the splits in $A$ and $B$ using red boxes, and then show that they have common split (which splits the overarching CSS code) in green.\\

\begin{figure*}[t]
\centering
\includegraphics{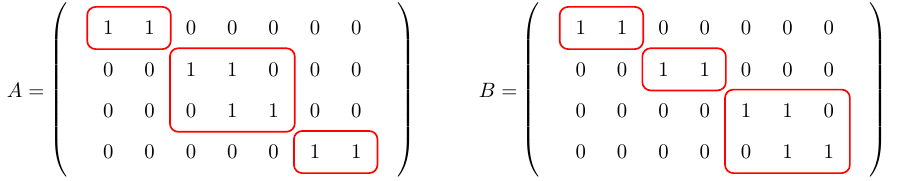}
\includegraphics{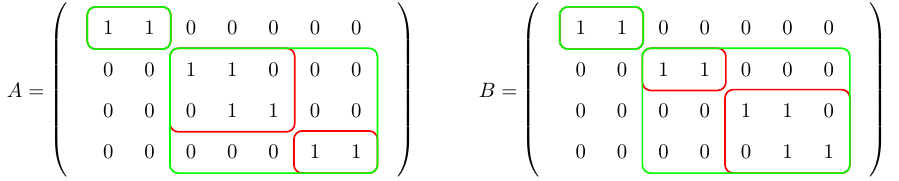}
\caption{Visualization of the splitting using matrices}
 \label{matrix_split}
\end{figure*}
\end{example}

\begin{figure}[h]
    \centering
    \begin{minipage}{0.5\textwidth}
        \centering
        \includegraphics[width=1\textwidth]{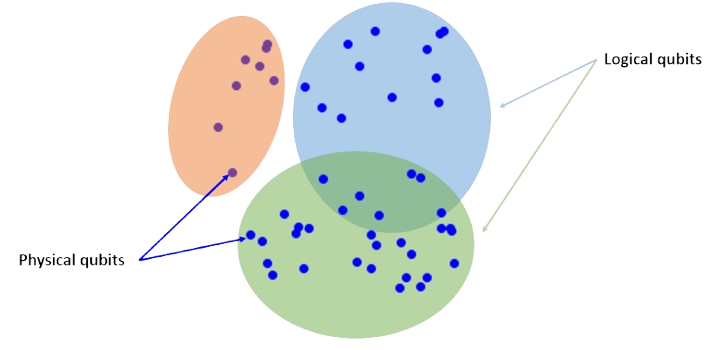}
        \caption{Visualization of a splitting code}
        \label{fig:splitting-code}
    \end{minipage}\hfill
    \begin{minipage}{0.5\textwidth}
        \centering
        \includegraphics[width=1\textwidth]{images/non_splitting_code.PNG}
        \caption{Visualization of a non-splitting code}
        \label{fig:non-splitting-code}
    \end{minipage}
\end{figure}

We provide the following proposition, whose proof is straightforward:

\begin{proposition}\label{prop:split-code-params}
    Let $\Code$ be a $\llbracket n,k,d\rrbracket$ CSS code that splits into an $\llbracket n_1,k_1,d_1\rrbracket$-code $\Code_1$ and an $\llbracket n_2,k_2,d_2\rrbracket$-code $\Code_2$. Then:
    \begin{itemize}
        \item $\frac{k}{n}=\frac{k_1+k_2}{n_1+n_2}\leq \max_{i\in\{1,2\}}\frac{k_i}{n_i}$.
        \item 
        $d=\min_{i\in\{1,2\}}d_i$
        \item 
        If $U$ is addressable on $\Code_1$ and $\Code_2$ then $U$ is addressable on $\Code$. If $U$ is partially addressable on $\Code_1$ or $\Code_2$ then $U$ is partially addressable on $\Code$.
    \end{itemize}
\end{proposition}

These results tell us that if a code splits, the subcodes cannot have worse rates. Thus, for an asymptotically good code, it must have some non-splitting, ``irreducible'' core. We also know this core must encode more than 1 logical qubit, as the rates of single-logical qubit codes cannot be arbitrarily high.

\section{Clifford Addressability}\label{sec:cliffords}
In this section we study $1$-local Clifford circuits: circuits made only from single-qubit Clifford gates. We will show that for circuits with $H$ gates, even being a logical operator at all (let alone what action it might have) implies the code splits. Hence, for good (non-splitting) CSS codes they cannot have such logical operators.

\subsection{Tools}
\begin{proposition}\label{prop:clifford logical}
    Let $U$ be a Clifford circuit and let $\mathcal{S}$ be the stabilizer group of $\Code$. If $U$ is a logical operator, then $U\mathcal{S}U^\dagger\subseteq \mathcal{S}$.
\end{proposition}

\begin{proof}
    Using \cref{cond_valid_unitary} we get that $U$ must send logical identities to logical identities. Since $U$ is Clifford, it sends Paulis to Paulis. Since $\mathcal{S}$ corresponds to all Pauli logical identities, we get $U\mathcal{S}U^\dag \subseteq \mathcal{S}$. 
\end{proof}

In particular, for a CSS code $\overline{\CSS}(A,B)$, the condition in \cref{prop:clifford logical} can be reformulated as: for all $a\in A$ and $b\in B$ (including $0$ for either)  \[ UX^aZ^bU^\dagger = \pm X^{a'}Z^{b'} \] where $a'\in A$, $b'\in B$.

\begin{definition}\label{def:equivalence of cliffords}
    Let $U,V$ two $1$-local Clifford circuits. We say that $U$ is equivalent to $V$ when there exists a Pauli string $P_{U,V}$ such that $U = VP_{U,V}$. 

\end{definition}

The following Lemma gives two useful properties which we will explain after its proof.

\begin{lemma}\label{lem:pauli-equivalence}
    Let $\Stab$ be the Pauli stabilizers of a CSS code and let $U$ be a $1$-local Clifford circuit. Then the following are equivalent: 
    \begin{enumerate}
        \item $U\Stab U^\dagger \subseteq \pm \Stab$ 
        \item $V\Stab V^\dagger \subseteq \pm \Stab$ for all $V$ equivalent to $U$
        \item There exists $V$ equivalent to $U$ such that $V\Stab V^\dagger \subseteq \Stab$.
    \end{enumerate}
\end{lemma}
\begin{proof}
    (1)$\Rightarrow $(2): Suppose $V=UP$ for some Pauli string $P$. Then for any $S\in\Stab$, $VSV^\dagger = U(PSP^\dagger)U^\dagger$. We know that $PSP^\dagger =(-1)^x S$ for some $x$, as $P$ and $S$ are both Paulis; since $U SU^\dagger = (-1)^yS$ for some $y$, we have $VSV^\dagger = (-1)^{x+y}S$, so $V\Stab V^\dagger\subseteq \pm \Stab$.

    (1)$\Rightarrow $(3): Let $\{X^{a_i}\}$ and $\{Z^{b_i}\}$ be generating sets for $\Stab$; equivalently, $\{a_i\}$ and $\{b_i\}$ are bases for $\Code_X^\perp$ and $\Code_Z^\perp$, respectively. Define $\{x_i\}$ and $\{y_i\}$ as follows: If $U X^{a_i}U^\dagger\in \Stab$, let $x_i=0$; otherwise $x_i=1$; if $U Z^{b_i}U^\dagger\in \Stab$ then $y_i=0$ and otherwise $y_i=1$.

    Given these, we can see that there are solutions $u,v\in\{0,1\}^n$ (if the code has $n$ qubits) to the linear system $a_i\cdot u = x_i$ and $b_i\cdot v=y_i$ for all $i$. Let $V = UX^{v}Z^{u}$. We have:
    \begin{align}
        V X^{a_i}V^\dagger =& UX^vZ^u X^{a_i}Z^uX^v U^\dagger\\
        =& U(-1)^{x_i}X^{a_i}U^\dagger\in \Stab
    \end{align}
    by anti-commutation of $X$ and $Z$. The same logic applies to all $Z^{b_i}$, and since these are generating sets, we conclude $V \Stab V^\dagger \subseteq \Stab$.

    (2) implies (1) trivially, which in turn means that (3) implies that all $V'$ equivalent to $V$ preserve $\Stab$ up to sign, and since $U$ is equivalent to $V$, this gives (1).
\end{proof}

\Cref{lem:pauli-equivalence} shows that if we are asking whether a Clifford circuit preserves a code, either its entire equivalence class preserves the code up to phase, or nothing equivalent to it will. Thus, we can focus on whichever equivalent Clifford is most convenient for our proofs.

Second, we see that we can ignore phase, as we can always fix up the phase with some circuit of Pauli gates (which are presumably easy to implement).

\begin{proposition}\label{prop:small-clifford-set}
    Let $U$ be a $1$-local Clifford circuit. Then there is a depth-1 circuit $V$, with gates from the set $\mathcal{C}_{S,H}:=\{I,S,HS,SHS,SH,H\}$, such that $U\mathcal{S}U^\dagger \subseteq \pm \mathcal{S}$ iff $V\mathcal{S}V^\dagger \subseteq \pm \mathcal{S}$.  

    Furthermore, writing $\Bar{U},\Bar{V}$ the logical action of these Clifford circuits, we have $\Bar{U} = \pm \Bar{V}$.
\end{proposition}
\begin{proof}
    By definition, on each qubit $U$ is a product of gates in $\mathcal{C}_{S,H}$ with Paulis. Since the gates in $\mathcal{C}_{S,H}$ preserve Paulis by conjugation, we can commute all the Paulis to the right. Thus, $U=VP$ where $P$ is a Pauli circuit and $V$ only contains gates from $\mathcal{C}_{S,H}$. By \cref{def:equivalence of cliffords}, we have that $U \equiv V$.

    Then since Paulis either commute or anti-commute, for any $X^aZ^b$, we have that 
    \begin{align}
        UX^aZ^bU^\dagger = & VPX^aZ^b P^\dagger V^\dagger\\
        = & \pm VX^aZ^bV^\dagger
    \end{align}
    This directly gives that $UX^aZ^bU^\dagger \in \pm \mathcal{S}$ if and only if $VX^aZ^bV^\dagger \in \mathcal{S}$.

    Furthermore, the logical action of $U$ is determined by how it acts on the Paulis $X^aZ^b$ with $a\in N(B)$ and $b\in N(B)$. Since $UX^aZ^bU^\dagger = \pm VX^aZ^bV^\dagger$, we have $\Bar{U} = \pm \Bar{V}$ where $\Bar{U},\Bar{V}$ denote the logical actions of $U,V$.
\end{proof}

Intuitively, \cref{prop:small-clifford-set} implies that any $1$-local Clifford operator can be decomposed into a circuit of $\mathcal{C}_{S,H}$ gates followed by a Pauli operator. Since we study stabilizer preservation up to a phase, we can ignore this Pauli factor in our analysis. Consequently, we need only consider the $\mathcal{C}_{S,H}$ component of the Clifford to derive our impossibility result.

Thanks to \Cref{lem:pauli-equivalence,prop:small-clifford-set}, \textbf{we will assume every $1$-local Clifford circuit has been compiled down to contain only $\{I,S,HS,SH,SHS,H\}$ (we take representative of the cosets), followed by Paulis.} Up to phase, those gates act on single-qubit Paulis as:

\begin{align}
I: & X\mapsto X,\hspace{-4em}& Z\mapsto Z\label{cliff:I}\\
S: & X\mapsto XZ, & Z\mapsto Z\label{cliff:P}\\
HS: & X\mapsto XZ, & Z\mapsto X\\
SH : & X\mapsto Z, & Z\mapsto XZ\\
SHS: & X\mapsto X, & Z\mapsto XZ\label{cliff:PHP}\\
H: & X\mapsto Z, & Z\mapsto X
\end{align}

\begin{definition}
    Let $U$ be a $1$-local Clifford circuit. For a set $K$ of specific single-qubit Cliffords, we define $U_K=\{i : U_i\in K\}$.
\end{definition}

\begin{example}
    Let $U=S\otimes XZH\otimes SH\otimes H$. Then this is equivalent by \Cref{prop:small-clifford-set} to $S\otimes H\otimes SH\otimes H$, and $U_H=\{2,4\}$ and $U_{S,H}=\{1,2,4\}$. To compare, $U_{S,H,SH}$ would give $\{1,2,3,4\}$.
\end{example}
This means that for any $1$-local Clifford circuit $U$, the equivalent $V$ can be written as $V=\bigotimes_{R\in \mathcal{C}_{S,H}}R^{V_R}$.

All of this leads to our main technical tool:
\begin{proposition}\label{prop:clifford-split-taxonomy}
    Let $\Code=\overline{\CSS}(A,B)$ and $U$ a Clifford circuit. If $U$ is a logical operator on $\Code$, then it is equivalent to some $V$ such that for all $a\in A$ and all $b\in B$:
    \begin{align}
        a\cap V_{S,HS,SH,H}&\in B\label{eq:cliff-ab}\\
        a\cap V_{SH,H}&\in A \label{eq:cliff-aa}\\
        b\cap V_{HS,SHS,SH,H}&\in A\label{eq:cliff-ba}\\
        b\cap V_{HS,H}&\in B\label{eq:cliff-bb}
    \end{align}
\end{proposition}
\begin{proof}
    Using $\sim$ as equality up to phase and $\equiv$ as equivalence up to stabilizer, we know that $UX^aU^\dagger$ must be a stabilizer for any $a\in A$. Thus:
    \begin{align*}
        UX^aU^\dagger &\sim VX^aV^\dagger\\
        &\sim X^{a\cap V_{I,SHS,S,HS}}Z^{a\cap V_{S,HS,SH,H}}\\
        &\equiv X^{a\cap V_{SH,H}}Z^{a\cap V_{S,HS,SH,H}}
    \end{align*}
    where the last equivalence holds since $X^a$ is a stabilizer and (since all gates in $V$ are in $\mathcal{C}_{S,H}$) we have that $a\cap U_{SH,H}+a\cap U_{I,S,HS,SHS} = a$.

    In a CSS code, if product $X^{a'}Z^{b'}$ is a stabilizer, then both $X^{a'}$ and $Z^{b'}$ are stabilizers, so $a'\in A$ and $b'\in B$. Applying this to the above gives the first two lines of the result; repeating the logic with a $Z$ stabilizer $Z^b$ completes the proof.

\end{proof}

\subsection{Impossibility Results for Hadamard-Type Circuits}\label{sec:impossibility_clifford}

Our first result follows almost immediately from \Cref{prop:clifford-split-taxonomy}. The strategy of the proof is as follows: first we identify what kind of circuits are valid logical implementation on a non-splitting CSS code (\cref{cor:splitting-h-types}). Secondly, we analyze the set of logical actions generated by these circuits (\cref{prop:no-partial-from-H-type}). For example, by showing that any targeted logical $H$ (or $HS, SH,CNOT$) is not in this set, we obtain impossibility result on the partial addressability (\cref{prop:no-h-cnot}). Finally, we extend these result to any code (not just non-splitting) by reducing the impossibility from partial addressability to addressability (\cref{cor:no addressability_local_clifford}). 

\begin{proposition}\label{cor:splitting-h-types}
    Let $\Code=\overline{\CSS}(A,B)$ and $U$ a $1$-local Clifford circuit. If $U$ is a logical operator, then the code splits in $U_{HS}$, $U_{SH}$, and $U_H$.
\end{proposition}
\begin{proof}
    If $a\in A$, then \Cref{eq:cliff-aa} tells us $a\cap U_{SH,H}\in A$. Applying \Cref{eq:cliff-ab} to $a\cap U_{SH,H}$ tells us $a\cap U_{SH,H}\in B$ as well, since $U_{SH,H}\subseteq U_{S,HS,SH,H}$. Then \Cref{eq:cliff-bb} applied to $a\cap U_{SH,H}$ tells us $(a\cap U_{SH,H})\cap U_{HS,H}=a\cap U_H\in B$. Then applying \Cref{eq:cliff-ba} to $a\cap U_H$, we see that $a\cap U_H\in A$ as well. Thus, $A$ splits on $U_H$. 

    Doing the same procedure for an arbitrary $b\in B$ shows that $b\cap U_H\in B$, so $B$ splits on $U_H$ as well, and hence the full CSS code splits.

    Similarly pushing these intersections through the statement of \Cref{prop:clifford-split-taxonomy} gives the other two results.
\end{proof}

Our results so far restrict the set of circuits that are logical operators, \emph{regardless of their action}. We now show that no addressable logical action is possible:
\begin{proposition}\label{prop:no-partial-from-H-type}
    Let $\Code$ be a non-splitting CSS code, and $U$ a $1$-local Clifford partially addressable unitary. Then the physical implementation of $U$ contains no $H$, $SH$, or $HS$ gates.
\end{proposition}
\begin{proof}
    \Cref{cor:splitting-h-types} tells us that we either apply $H$ to all or none of the qubits (similarly with $SH$ and $HS$). If our circuit $C$ for $U$ applies $H$ to every gate, then for any logical operator $X^{x_i}$, $CX^{x_i}C^\dagger = Z^{x_i}$. In particular, if $U$ is partially addressable, there must be some logical operator $X^{x_i}$ which should be unchanged by the action of $U$. But $Z^{x_i}$ cannot be equivalent by stabilizers to $X^{x_i}$, as $x_i$ is not in the span of $A$ (where the code is $\overline{\CSS}(A,B)$. 

    The same reasoning holds if the circuit contains only $SH$ or $HS$ gates, by taking their action on either a logical $X$ or logical $Z$ operator.

\end{proof}

Thanks to \Cref{prop:no-partial-from-H-type}, we know that $1$-local Clifford addressable gate must only use $S$ or $SHS$ gates. We now show that this set is too restrictive to enact any $H$, $SH$, $HS$, or $\CNOT$ gates.

\begin{proposition}\label{prop:no-h-cnot}
    The $H$, $SH$, $HS$ and $\CNOT$ gates are not $1$-local Clifford partially addressable for any non-splitting CSS code.
\end{proposition}
\begin{proof}
    Using \Cref{prop:no-partial-from-H-type}, no matter how we construct our circuit $U$, if it enacts some addressable logical action it must be equivalent to $U$ made from only $I$, $S$, and $SHS$ gates. For any logical $X$ operator $X^x$, or logical $Z$ operator $Z^z$, we can see that the action of $U$ on these will be 
    \begin{align}
        UX^xU^\dagger = & X^xZ^{x\cap U_S}\\
       \text{ and } UZ^zU^\dagger = & X^{z\cap U_{SHS}}Z^z
    \end{align}
    by the commutation rules in \Cref{cliff:P} and \ref{cliff:PHP}. 

    Throughout, we will use $X^{x_i}$ to denote a logical $X$ operator on the $i$th logical qubit, and the same for $Z^{z_i}$.

    If we want $U$ to enact an addressable $\overline{H}$ or $\overline{SH}$ gate, there must be some qubit $i$ such that $U$ sends $X^{x_i}$ to an operator equivalent to $Z^{z_i}$. Thus, $X^xZ^{x\cap U_S}$ must be equivalent up to stabilizers to $Z^{z_i}$. This would imply $X^x$ is a stabilizer, contradicting that it was a logical $X$ operator. The same logic applied $Z^{z_i}$ works to show that $\overline{HS}$ cannot address a qubit $i$.

    For $U$ to enact an addressable $\overline{\CNOT}$ gate, it must act on some target $i$ and control $j$, so that $UX^{x_i}U^\dagger\equiv X^{x_i+x_j}$. This would imply $X^{x_i}Z^{x_i\cap U_S}\equiv X^{x_i+x_j}$, but this would mean $X^{x_j}$ is  a stabilizer, not a logical operator (a contradiction).

\end{proof}

As a final note, to extend our result to general codes, we show how our non-splitting requirement implies a distance bound. The idea is quite natural: if a code has $1$-local Clifford addressable $H$ or $HS$ or $SH$, then by previous results it splits into a collection of single logical qubit codes. This bounds the performance of this code, by the performance of a single logical qubit code. In particular, any code breaking this bound does not admit $1$-local Clifford addressable $H$, $SH$, $HS$, or $\CNOT$.

\begin{corollary}\label{cor:no addressability_local_clifford}
    If an $\llbracket n,k,d\rrbracket$ CSS code $\Code$ admits a $1$-local Clifford addressable $H$, $SH$, $HS$, or $\CNOT$ gate, then its rate is at most $\frac{1}{2d+1}$.
\end{corollary}
\begin{proof}
   By \Cref{prop:no-h-cnot}, such a code must split, and the subcodes must further split until they encode only 1 logical qubit. Let the parameters of each code be $\llbracket n_i,1,d_i\rrbracket$. The quantum singleton bound \cite{PRA:CerCle1997} tells us that $n_i -1\geq 2(d_i-1)$. We know $\sum_{i=1}^k n_i = n$, meaning $n-k \geq 2(\sum_{i=1}^kd_i-1) \geq 2(k\min_i\{d_i\}-1)$. However, $\min_i\{d_i\}=d$ by \Cref{prop:split-code-params}, so $n-k \geq 2kd-2$. Rearranging gives $\frac{k}{n} \leq \frac{1}{2d+1} - \frac{2}{k}$.

\end{proof}

\subsection{The Clifford Hierarchy}

The previous sections showed that H-type gates are forbidden for non-splitting codes, leaving $S$ and $SHS$ as the only viable 1-local Clifford logical operators. This raises a natural question: what about gates higher in the Clifford hierarchy ($T$ gates, $CCZ$, and beyond) ? One might hope that moving to higher levels unlocks new addressable operations. We show the opposite: if $S$ and $SHS$ gates admit no logical identities on the code, then no 1-local gate from any higher level of the hierarchy can even preserve the codespace (\cref{cor:P-or-nothing-Clifford-hierarchy}). The argument is inductive. A gate at level $k+1$ of the Clifford hierarchy conjugates Pauli stabilizers into level-$k$ operators. If the resulting operator must fix all codewords, it would be a level-$k$ logical identity, contradicting the inductive hypothesis (\cref{thm:no-clifford-hierarchy-stabilizers}). Thus the entire hierarchy collapses: either $S$ can be used to obtain logical identities, or nothing above it is even a logical operator. The $S$ gate is not merely one gate among many, it is the gateway to the rest of the Clifford hierarchy.

\begin{theorem}\label{thm:no-clifford-hierarchy-stabilizers}
    Let $\Code$ be a non-splitting CSS code of distance at least 2. Suppose there is some $k$ such that no 1-local circuit in $\mathcal{C}_n^k\setminus \mathcal{C}_n^{k-1}$ is a logical identity. Then for all $m\geq k$, no 1-local circuit in $\mathcal{C}^m\setminus\mathcal{C}^{k-1}$ can preserve the codespace.
\end{theorem}
\begin{proof}
     We show this inductively, noting that if we prove that no circuit preserves the codespace, this also implies no circuit can be a logical identity.
    
    The base case is given by the assumptions of the theorem, so we suppose that there are no 1-local stabilizers with gates from the $m$ level of the Clifford hierarchy, for some $m$.

    Suppose $U$ is a circuit composed of 1-local gates in the $m+1$th level of the Clifford hierarchy, such that at least one physical qubit $i$ has a gate which is in the $m+1$th level but not the $m$ level, and such that $U$ preserves the codespace. That is, 

    \begin{equation*}
        U = U_1 \otimes U_2\otimes \dots \otimes U_n
    \end{equation*}
    where each $U_n$ is a single-qubit gate in $\mathcal{C}_1^{m+1}$, and $U_i\notin\mathcal{C}_1^{m}$.

    Because $U_i\notin\mathcal{C}_1^{m}$, there is a Pauli matrix $V$ (either $X$, $Z$, or $XZ$) such that $U_iV=A_iU_i$ for $A_i\in \mathcal{C}_1^{m}\setminus\mathcal{C}_1^{m-1}$. Because the code has distance greater than 1 and does not split, for every physical qubit there is both an $X$ and $Z$ stabilizer with support on that qubit, and thus there is some Pauli stabilizer $P_0$ whose support on $i$ is the Pauli matrix $V$ defined above. 
    
    Since $U$ is in the $m+1$th level, there is a operator $A$ such that $UP_0=AU$, where $A=A_1\otimes\dots\otimes A_n\in \mathcal{C}^{m}$. By choice of $U$, we know that the action of $A$ on the $i$th qubit, $A_i$, is not in the $m-1$ level of the Clifford hierarchy, so $A\in \mathcal{C}_n^{m}\setminus\mathcal{C}_n^{m-1}$ as well.
    
    Then we have that for any state $\ket{\psi'}$ in the code, there must be some $\ket{\psi}$ in the code such that $U\ket{\psi'}=\ket{\psi}$. This tells us:
    \begin{align*}
         \ket{\psi} = & U\ket{\psi'}\\
         = & UP_0\ket{\psi'}\\
        =&AU\ket{\psi'}\\
        = & A\ket{\psi}
    \end{align*}
    Thus, $A$ is a logical identity on the codespace. This contradicts the inductive hypothesis since $A$ is a 1-local circuit in $\mathcal{C}_n^{m}\setminus\mathcal{C}_n^{m-1}$.

    Thus, the assumption that $U$ exists must be false.
\end{proof}

\cref{thm:no-clifford-hierarchy-stabilizers} links logical identities at level $k$ to logical operators at level $k+1$: if no 1-local circuit at level $k$ acts as a logical identity, then no 1-local circuit at level $k+1$ can preserve the codespace. The chain therefore only needs a base case, which is precisely what \cref{sec:impossibility_clifford} 
restricts: very few implementations using $H,HS,SH$ can implement logical identities. \cref{cor:P-or-nothing-Clifford-hierarchy} combines these two ingredients, identifying $S$ and $SHS$ as the pivotal gates: if neither gives rise to a logical identity, the entire hierarchy above is ruled out.

\begin{corollary}\label{cor:P-or-nothing-Clifford-hierarchy}
    If a CSS code $\Code=\overline{\CSS}(A,B)$:
    \begin{itemize}
        \item is non-splitting;
        \item has distance at least 2;
        \item is not self-dual (i.e., $A\neq B$);
        \item admits no logical operator from $S$ or $SHS$ gates;
    \end{itemize}
   then $\Code$ does not admit any circuit from 1-local gates in any higher level of the Clifford hierarchy.
\end{corollary}
\begin{proof}
    Using \Cref{thm:no-clifford-hierarchy-stabilizers} we only need to show that such a code admits no logical identities from 1-local Clifford gates except the Paulis.

    By assumption, it admits no logical operator from $S$ or $SHS$, so they cannot be logical identities.

    Then we consider $H$, $HS$, and $SH$ gates. From \Cref{prop:no-partial-from-H-type}, we know that a circuit from these gates can only preserve the code if we apply the physical gate to every physical qubit. In this case, the circuit can only preserve the code if $A=B$, using the rules in \Cref{prop:clifford-split-taxonomy}, in which case the circuit has non-identity logical action. Thus $H$, $HS$, and $SH$ cannot form logical identities on $\Code$.

\end{proof}

In short, either the code admits logical operators from phase (or $SHS$) gates, or it admits nothing else from the (1-local) Clifford hierarchy. One should remark that this result is much more general than the previous ones as it is not about addressability, but about being a logical operator. 

In turn, we might wonder whether similar techniques from our $H$ impossibility results could apply to just $S$. However, this is not true: the same commutation rules from \Cref{cliff:P} tell us that if a code  $\overline{\CSS}(A,B)$ admits a logical operator formed by single-qubit $S$ circuit on a set of physical qubits $p$, then we can conclude that $a\cap p\in B$ for all $a\in A$. This is significantly less restrictive than the requirements for $H$, and indeed codes with transversal T gates (such as \cite{PRA:BraHaa2012}) admit partially addressable $S$ gates. One can see this by noting that for $T^t$ to preserve the codespace, by \Cref{cond_valid_unitary}, $T^tX^aT^{-t}$ must be a logical identity for the code for any stabilizer $X^a$, but this operator will include $S$ gates on qubits in $a\cap t$.

\section{Permutation Isomorphisms and Addressability}\label{sec:isomorphisms}
Because of our restrictions that addressable gates should be logical operators, i.e., preserve the codespace, we are effectively studying automorphisms of the code. Here we consider this more directly, but consider isomorphisms between codes, and consider only isomorphisms formed by permutations of qubits. Such isomorphisms (in particular automorphisms) were proposed as a method to perform logical operations as early as 2013~\cite{markus}, and recently detailed for certain quantum LDPC codes~\cite{photonic2025swapCNOT}. 

We show that there are actually not that many isomorphisms like this, which rules out certain kinds of circuits for high-rate codes. Moreover, after parsing which circuits are valid implementation, we further reduce the number of implementations by showing that many of them will have the same logical action. Hence, we obtain an upperbound on the number of different logical actions that isomorphisms can implement (\cref{thm:logical-permutation-count}). Finally, we conclude that on codes with high rate, this number is not big enough to possibly have addressable SWAPs or CNOTs (\cref{coro style}). 

Throughout this section we will implicitly work with \emph{qudit} CSS codes except where otherwise mentioned, hence the use of $\mathbb{F}_q$ instead of $\mathbb{F}_2$.

\subsection{Counting Permutations}

\begin{definition}
    Let $r,n \in \mathbb{N}$, and $\Code_1,\Code_2$ be classical codes with $r\times n$ parity check matrices $H_1,H_2 \in \mathcal{M}_{r \times n}(\mathbb{F}_q)$, we say that $\tau_n \in S_n$ is a \emph{permutation isomorphism} from $\Code_1$ to $\Code_2$ iff there exists $U \in \text{GL}_r(\mathbb{F}_q)$ such that $UH_1 = H_2P$ where $P$ is the permutation matrix of $\tau_n$.
    
    If $\Code_1=\Code_2$, we say $\tau_n$ is a \emph{permutation automorphism}.
\end{definition}

\begin{remark}
    We can define it equivalently with the generator matrix, since if $P$ preserve the codespace, it equivalently preserves its orthogonality.
\end{remark}

\begin{proof}
    Suppose $UG_1=G_2P$. For any $x\in \Code_2^\perp$ and $y\in\Code_1$, $\langle y,xP\rangle=\langle Uy',xP\rangle$ for some $y'\in \Code_1$, since $U$ is invertible. But this means $Uy'=y''P$ for some $y''\in \Code_2$, so the inner product is equal to $\langle y''P,xP\rangle = \langle y'',x\rangle=0$. Thus, $xP\in \Code_1^\perp$ for any $x\in \Code_2^\perp$; by dimensionality arguments, this tells us $UH_1=H_2P$.

\end{proof}

Since CSS codes are built from classical codes, this definition naturally extends to CSS codes:

\begin{definition}
    $\tau_n \in S_n$ is a permutation isomorphism from $\CSS(\Code_{1,X},\Code_{1,Z})$ to $\CSS(\Code_{2,X},\Code_{2,Z})$ iff it is a permutation isomorphism from $\Code_{1,X}$ to $\Code_{2,X}$ and from $\Code_{1,Z}$ to $\Code_{2,Z}$.

\end{definition}
  In this case, for $i\in\{1,2\}$ we can form a matrix $H_i = \begin{pmatrix}
        H_{i,X}^T & H_{i,Z}^T
    \end{pmatrix}^T$ where $H_{i,X},H_{i,Z}$ are the respective parity checks of $\Code_{i,X},\Code_{i,Z}$. Then we get that 
    \begin{equation}
    UH_1 = H_2P\text{ with }U = \begin{pmatrix}
        U_X & 0 \\ 0 & U_Z
    \end{pmatrix}.
    \end{equation}

\begin{example}
    For an automorphism, take $H_1 = \begin{pmatrix}
        1&1&0\end{pmatrix} $ and $H_2 = \begin{pmatrix}
            0&1&1
        \end{pmatrix}$, 
    we get $H = \begin{pmatrix}
            1&1&0\\
            0&1&1
    \end{pmatrix}$. \\
    
    We can see that 
    \[\begin{pmatrix}
        0&1\\1&1
    \end{pmatrix}\begin{pmatrix}
            1&1&0\\
            0&1&1
    \end{pmatrix} = \begin{pmatrix}
            1&1&0\\
            0&1&1
    \end{pmatrix}\begin{pmatrix}
            0&1&0\\
            0&0&1\\
            1&0&0
    \end{pmatrix}.\]
    
This means that the permutation of qubits represented by the permutation matrix on the right preserves the stabilizer group, meaning that it is a valid permutation automorphism for the classical code represented by $H$ but it is not a valid permutation automorphism for $\CSS(\mathcal{C}_1,\mathcal{C}_2)$ as it is not an automorphism on $\mathcal{C}_1,\mathcal{C}_2$.
\end{example}

\begin{remark}
The matrix $H = \begin{pmatrix}
    H_1^T & H_2^T
\end{pmatrix}^T$, corresponds to the parity check matrix of the CSS code where we forget if a check is of type X or Z.
\end{remark}

Intuitively, this means that a permutation isomorphism can permute the columns of $G$ from the first code in such a way that it maps to the second codespace. However this puts strong conditions on the form of the invertible matrix $U$. From those conditions we can extract an upper bound on the number of different such matrices $U$, and thus and upper bound of the number of permutation automorphisms.

\begin{proposition}\label{equality_carnical}
    The number of pairs $U,P$ of invertible matrix and permutation such that $UH_2 = H_1P$ only depends on the vector spaces $\Span(H_1)$ and $\Span(H_2)$. Moreover, the number of such pairs is constant under permutations of the columns of $H_2$ or $H_1$. 
\end{proposition}

\begin{proof}
    Let $\Bar{H_1}$ be another basis of the first vector space and $\Bar{H_2}$ another basis of the second. This means that there exists $W_1$ and $W_2$ such that $\Bar{H_1} = W_1H_1$ and $\bar{H_2}=W_2H_2$. Let $P'_1$ and $P'_2$ be arbitrary permutations.

    Let $U$ be an invertible matrix and let $P$ be a permutation such that $UH_2=H_1P$, or $UH_2P^{-1}=H_1$. Then $(W_1UW_2^{-1})\Bar{H}_2P'_2(P_2^{'-1}P^{-1}P'_1) = W_1H_1P'_1= \Bar{H}_1P'_1$. Since $WUW^{-1}$ is invertible and $P_2^{'-1}P^{-1}P'_1$ is a permutation, this is a valid pair for $\Bar{H}_2P'_2$ and $\Bar{H}_1P'_1$ Thus, there are as many pairs for $(H_1,H_2)$ as for $(\Bar{H}_1P'_1,\Bar{H}_2P'_2)$.

\end{proof}

Here we make a distinction between counting distinct pairs $(U,P)$, counting distinct permutations $P$ that belong to such a pair, and counting distinct $U$ that belong to such a pair. 

\begin{lemma}\label{lem:count-u-and-pairs}
    Let $m$ be the number of distinct pairs $(U,P)$ of an invertible matrix $U$ and a permutation $P$ such that $UH_2=H_1P$, and let $m_u$ be the number of distinct $U$ from such a pair and $m_p$ the number of distinct $P$. Then for a fixed $U$, if it is part of pair $(U,P_1)$ then $(U,P_2)$ is also a valid pair if and only if $P_2$ is in the same right coset as $P_1$ of the subgroup of permutations $\text{Sym}(H_1):=\{P\in\text{Sym}(n): H_1P=H_1\}$. Consequently, $m_p=m=m_u\vert\text{Sym}(H_1)\vert$.
\end{lemma}
\begin{proof}
    To prove $m_p=m$, we prove that each $P$ has exactly one $U$ that forms such a pair. If there were more than one, then $U_1H_2=H_1P=U_2H_2$, meaning $U_2U_1^{-1}H_2=H_2$. Since $H_2$ is full-rank, this implies $U_2U_1^{-1}=I$. 

    For the second, if $(U,P_1)$ and $(U,P_2)$ are valid pairs, then $UH_2=H_1P_1=H_1P_2$. This implies $H_1P_1P_2^{-1}=H_1$. Immediately this implies $P_1P_2^{-1}\in\text{Sym}(H_1)$. Conversely, if any $P$ is in $\text{Sym}(H_1)P_1$, then $P = P'P_1$ where $P'\in\text{Sym}(H_1)$. Then $H_1P=H_1P'P_1=H_1P_1$, so $UH_2=H_1P$ as well.

\end{proof}

The size of this group of permutations that fix $H_1$ can be easily found, since they are simply products of permutations of columns of $H_1$ which are exactly the same. 

Clearly the number of valid permutations $P$ tells us the number of qudit permutations which preserve the code. However, we will actually care more about the number of invertible matrices $U$, which is smaller. We need a small lemma to prove that we can quotient out by $\Sym(H_1)$ in this way:

\begin{lemma}\label{lem:duplicate-columns}
    Let $\Code=\overline{\CSS}(A,B)$, with $H_A$ a parity check matrix for $A$. Then if two columns of $H_A$ are identical, either $\Code$ has distance at most 2 or swapping these two qubits acts as a logical identity.
\end{lemma}
\begin{proof}
    Suppose w.l.o.g. the first two columns of $H_A$ are the same. Then the vector $v:=(1,-1,0,\dots, 0)\in A^\perp$, so $Z^{v}$ is either a logical $Z$ operator or a $Z$ stabilizer. If $Z^v$ is a logical operator, the code has distance at most 2. 

    If $Z^v$ is a stabilizer, then every $x\in B^\perp\subseteq \F_q^n$ must have $x_1=x_2$ to be orthogonal to $v$. However, all logical $X$ operators must be of the form $X^x$ for some $x\in B^\perp$, so $x_1=x_2$ for all $X$ stabilizers \emph{and} logical operators. Thus, swapping these two qubits acts as the identity on $X$ operators. 

    For $Z$ operators, take any $b=(b_1,b_2,\dots, b_n)\in A^\perp$ (either a $Z$ stabilizer or logical operator). If $b_1=b_2$ it is also invariant under swapping qubits 1 and 2, and if $b_1\neq b_2$, then swapping the first two qubits will produce the operator:
    \begin{equation}
        (b_2,b_1,b_3,\dots,b_n)=b + (b_2-b_1,b_1-b_2,0,\dots,0)
    \end{equation}
    However, the second vector is clearly a multiple of $v=(1,-1,0,\dots,0)$, and we know that $Z^v$ is a $Z$-stabilizer. Thus, the operator $Z^b$ after the swap is equivalent to itself up to multiplication by a stabilizer, so it must be the same logical operator.

    Both arguments also show that the $X$ and $Z$ \emph{stabilizers} are also invariant when $Z^v$ is a $Z$-stabilizer.

\end{proof}

\begin{corollary}\label{cor:identity-automorphisms}
    Let $\Code=\overline{\CSS}(A,B)$ have distance at least 3. Then any permutation automorphism $\tau$ such that $H_AP_\tau=H_A$ acts as a logical identity. 
\end{corollary}
\begin{proof}
    We can see that $H_AP_\tau=H_A$ only if it permutes only columns of $H_A$ which are identical. Thus, it can be constructed as a product of transpositions on such columns. From \Cref{lem:duplicate-columns}, each transposition is a logical identity, so their product is as well.

\end{proof}

Thus, to count logical operations, we care only about the number of invertible matrices $U$ such that $UH_2=H_1P$, not the number of pairs or permutations.

\begin{proposition}\label{upp nbr permutation}
    Calling $m$ the number of invertible matrices $U$ such that there exists a permutation $P$ such that $UH_2 = H_1P$, and $p$ the number of pairs of permutations such that $P'H_2' = H_1'P''$ we have 
$ p \leq m \leq \frac{n!}{(n-r)!}$, where $H_1'$ and $ H_2'$ are the right block in the row reduced forms of $H_1$ and $H_2$. 
\end{proposition}

\begin{proof}
     Since $H_2$ has $r$ independent rows, we can reduce it such that it has the form $\Bar{H}_2 = \left(\begin{array}{c | c}
        I_r & H_2' 
    \end{array}\right) = WH_2P'$ where $W \in GL_r(\mathbb{F}_q)$ and $P'$ a permutation matrix of dimension $n$.\\ 

    Using \cref{equality_carnical} we get that there are as many pairs for $H_2$ and for $\Bar{H}_2$. 

    Let $U$ be an invertible matrix of rank $r$, we get that $U\Bar{H}_2 = \left(\begin{array}{c | c}
        U & UH_2'
    \end{array}\right)=H_1P$. Thus, for $U$ to be part of a valid pair, it has to be made of $r$ independent columns of $H_1$. And there are less than $\frac{n!}{(n-r)!}$ ways to pick such set of columns. \\

    For the lower bound, we can check that we can extend any valid permutation on $H_1',H_2'$ into a valid permutation on $\Bar{H}_1$. Let $P',P''$ be permutation matrix of dimension respectively $r,n-r$.
    \begin{align*}
     P'\Bar{H}_2 &= \left(\begin{array}{c | c}
        P' & P'H_2'
     \end{array}\right) \\ 
                 &= \left(\begin{array}{c | c}
        P' & H_1'P''
     \end{array}\right) \\ 
                 &= \left(\begin{array}{c | c}
        I_r & H_1'
     \end{array}\right) \begin{pmatrix}
        P' & 0 \\ 
        0 & P''
    \end{pmatrix} \\ 
                 &= \Bar{H}_1 \begin{pmatrix}
        P' & 0 \\ 
        0 & P''
    \end{pmatrix}  
    \end{align*}

    Hence, any pair of permutations on $H_1',H_2'$ can be extended into a valid one on $\Bar{H}_1,\Bar{H}_2$.

\end{proof}

\begin{remark}
    Furthermore, since we want $U\Bar{H}_2 = \Bar{H}_1P$ with $P$ a permutation of the columns, then we also need $U\Bar{H}_2$ to generate $I_r$, meaning that the columns of $U^{-1}$ should also be in $\Bar{H}_2$. This is a stricter restriction and could reduce the upper bound significantly depending on the code. We can use this observation to create an algorithm to find all matrices $U$ having $UB = BP$ without going over all permutations. This will be slightly better but still very inefficient.  
\end{remark}

Finally, we summarize as a theorem:

\begin{theorem}\label{thm:logical-permutation-count}
    Let $\Code=\overline{\CSS}(A,B)$ be an $\llbracket n,k,d\rrbracket$ code with $d\geq 3$. The number of distinct logical operations that can be implemented by permuted qubits in the code is upper-bounded by $\frac{n!}{k_{max}!}$, where $k_{max}=\max\{\dim(A),\dim(B)\}$.
\end{theorem}
\begin{proof}
    From \Cref{cor:identity-automorphisms}, if two permutations $P_1$ and $P_2$ are such that $H_AP_1=H_AP_2$, they must produce the same logical action: this identity tells us $P_1=PP_2$ where $H_A=H_AP$, and $P$ acts as a logical identity. 

    Then \Cref{lem:count-u-and-pairs} tells us that the number of permutation automorphisms on the code, quotiented out by this subgroup, is precisely the number of invertible matrices $U$ where there is some permutation such that $UH_A=H_AP$.

    Then from \Cref{upp nbr permutation}, the number of such invertible matrices is at most $\frac{n!}{(n-\dim(H_A))!}$. Since $H_A$ is the parity check matrix, the dimension of the code $A$ is $k_A=n-\dim(H_A)$. Thus, there are at most $\frac{n!}{k_A!}$ distinct logical operations.

    The same logic must apply to $B$, so our upper bound is 
    \begin{equation}
        \min\left\{\frac{n!}{k_A!},\frac{n!}{k_B!}\right\}= \frac{n!}{k_{max}!}
    \end{equation}.

\end{proof}

\Cref{thm:logical-permutation-count} gives a bound in terms of the rates of the underlying \emph{classical} codes. Transforming this into a bound with a rate on the quantum code takes some care, using \Cref{prop:code-rates}.

\subsection{Operators Constructed From SWAPs}\label{sec:swap-operators}
We are now going to use the upper bound on the number of automorphisms on classical codes, to describe families of CSS codes on which certain gates are not SWAP addressable. The idea is that if we have more possible logical gates (e.g. SWAPs) than automorphisms, then SWAP cannot be SWAP addressable. 

We prove this using the lemmas in the Appendix, which describe the asymptotic behavior of the CSS code rates required for the theorem.

\begin{theorem}\label{thm no swap by swap}
    Let $\mathcal{C}_n = \CSS(C_n^1,C_n^2 )$ be a family of CSS codes such that calling $\rho_n$ the rate of $\mathcal{C}_n$ and $\rho'_n$ the maximum rate of the classical codes $C_n^1,C_n^2$, if there exists $n_1 \in \mathbb{N}$ such that for all $n>n_1$ we have $\rho_n + \rho_n' > 1$, then SWAP is not SWAP-addressable on this family of codes.    
\end{theorem}

\begin{proof}
    By \Cref{thm:logical-permutation-count}, the maximum number of distinct logical operations formed by permutations is at most $\frac{n!}{(\rho'_n n)!}$. 

    Since by assumption $\rho_n + \rho'_n > 1$, we can use \cref{integral_pas_beau} and conclude there exists some $n_0$ such that for all $n\geq n_0$, $\frac{n!}{(\rho'_n n)!} < (\rho_n n )! = k! $.  
    
    If the addressable SWAPs existed, we could compose them to obtain all $k!$ logical permutation gates from physical permutations; however, the inequality shows that there are not enough allowed physical permutation circuits to produce this many logical operators.

\end{proof}

\begin{remark}\label{remark cnots by swaps}
    We can swap two qubits by using 3 CNOTs between them:
    \begin{equation}
        \SWAP_{i,j} = \CNOT_{i,j} \CNOT_{j,i} \CNOT_{i,j}.
    \end{equation}
    This means that if we had all logical CNOTs then we could generate all logical SWAPs. Thus CNOT is not SWAP-addressable on those codes either. 
\end{remark}

Intuitively, if SWAP were addressable, we could compose addressable SWAPs to generate all $k!$ logical permutations of the $k$ logical qubits. However, \cref{thm:logical-permutation-count} tells us that permutations of physical qubits can only produce at most $\frac{n!}{k_{\max}!}$ distinct logical operations. For high-rate codes, this quantity grows much slower than $k!$, so there are simply not enough 
physical permutations to generate all logical ones. The same argument extends to CNOTs and any $2$-qubit gate, though with different rate threshold. This idea and the exact rate threshold are formalized in \cref{coro style}.

\begin{corollary}\label{coro style}
The following gates are not permutation addressable on CSS codes with the following rates:
    \begin{itemize}
        \item 
    SWAP gates for codes with an asymptotical rate greater than $\frac{1}{3}$;
    \item 
    Any 2-qubit gate for codes with an asymptotical rate greater than $\frac{3}{4}$;
    \item 
    CNOTs gates for codes with an asymptotical rate in $\Omega\left(\sqrt{\frac{\log n}{n}}\right)$.
    \end{itemize}
\end{corollary}
\begin{proof}
    Let $(\mathcal{C}_n)_{n \in \mathbb{N}}$ be a family of CSS codes and $n_0 \in \mathbb{N}$ such that $\forall \ n > n_0, \ \rho_n > \frac{1}{3}$.  Let us now fix $n > n_0$, and call $C_n^1,C_n^2$ the classical codes making $\mathcal{C}_n$, and $\rho_n',\rho_n''$ the maximum and minimum of their rates. We get that $\rho_n = \rho_n' + \rho_n'' -1$ by \Cref{prop:code-rates}. Thus $ \rho_n \leq 2\rho_n' -1 $ which means that $\rho_n' \geq \frac{\rho_n+1}{2}$. Hence $\rho_n + \rho_n' \geq \frac{\rho_n +1}{2} + \rho_n > \frac{4}{6} + \frac{1}{3} >  1$.
    Thus, for all $n>n_0$, $\rho_n + \rho_n' > 1$, and we can now use \cref{thm no swap by swap}. 

    For arbitrary 2-qubit gates, if a 2-qubit gate is permutation addressable it is parallel addressable since permutations are closed under composition. Thus, the number of such logical gates is at least the number of choices of disjoint pairs of qubits, and there are at least $\min\{k!!,k(k-1)!!\}$ such pairs.
    Applying the logic above with $\rho_n > \frac{3}{4}$, we see that $\frac{1}{2}\rho_n+\rho_n'>1$, and by \Cref{cor:count-double-factorial} this is greater than the number of permutation automorphisms on the code.

    For CNOT gates, CNOTs generate all invertible linear matrices on the computational basis. By \Cref{lem:count-GLn}, the size of invertible linear matrices on $k=\rho n$ logical qubits is asymptotically greater than $\frac{n!}{(\rho'n)!}$ when $\rho_n > \sqrt{\frac{\log n}{n}}+\Omega\left(\frac{1}{\sqrt{n}}\right)$. 

\end{proof}

These results illustrate an example of a trade-off between the performance of a code (its parameters) and how easy it might be to implement some logical operations on it. In \cite{quantum_tanner_codes,asymptotically_good_ldpc_PK} they prove methods to construct families of good quantum codes for any rate $0 < \rho < 1$; however, using \cref{thm no swap by swap} we know that using only physical swaps, none of these families can implement addressable CNOTs, and starting from $\rho > \frac{1}{3}$ they cannot implement addressable logical swaps.

In \cite{photonic2025swapCNOT} they construct various addressable Clifford gates, including CNOTs, from permutation automorphisms. The rate of their code family is $\Theta\left(\frac{\log^2 n}{n}\right)$, and \Cref{coro style} shows that they are only a quadratic factor away from our upper bound on the best possible rate with this technique.

\subsection{CNOT and CZ Results}\label{sec:automorphism-cnots}
Our techniques in the last section relied on counting automorphisms. We proved a more general statement about isomorphisms between codes so that we can consider addressable gates between two codes. We will consider two codes where we use physical CNOT or CZ gates, with the controls in one code and the targets in the second. We call a circuit \emph{global} if precisely one gate acts on every physical qubit. This condition captures the addressable CCZ gates in \cite{hvwz2025ccz}.

Thanks to this condition, the physical gates define a permutation $\pi\in \mathcal{S}_n$, where if $i$ is the control of a CNOT in the first code, $\pi(i)$ is the target of that CNOT in the second code (defined similarly for CZ). This representation using permutations makes it possible to link the results on isomorphisms to CNOTs and CZs. We will write such a circuit as $\CNOT(\pi)$ or $\CZ(\pi)$.

\begin{example}
    Consider the following unitary : 
    
    \begin{center}
    \resizebox{0.4\textwidth}{!}{
        \rotatebox{270}{
        \includegraphics[]{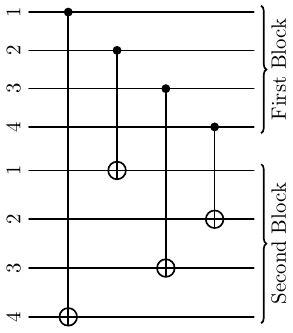}
        }
    }
    \end{center}
    
    where control and targets are from two different blocks of the same type of code made of $4$ physical qubits. 
    The permutation $\pi$ would here be $\pi = (142)(3)$. 
\end{example}

A critical component of the previous section was counting automorphisms \emph{up to distinct logical actions}. We prove an analogous result here. For a general gate $U$, we use $\CU(\pi)$ to denote a global controlled-$U$ circuit, where every qubit $i$ in one code is a control, whose target qubit is $\pi(i)$ in a second code.

\begin{lemma}\label{lem:cu-perm-invariant}
    Let $U$ be a single-qubit gate, and consider a global controlled-$U$ circuit with all controls in $\Code_1=\overline{CSS}(A_1,B_1)$ and all targets in $\Code_2=\overline{CSS}(A_2,B_2)$, and let $\pi$ be the permutation induced by mapping control to target. If $\pi_1\in\text{Sym}(H_{A_1})\cup \text{Sym}(H_{B_1})$ and $\pi_2\in\Sym(H_{A_2})\cup \Sym(H_{B_2})$, then $\CU(\pi_2\circ \pi \circ \pi_1)$ has the same logical action as $\CU(\pi)$ if both $\Code_1$ and $\Code_2$ have distance at least 3.
\end{lemma}
\begin{proof}
    We can see that $\CU(\pi_2\circ\pi\circ \pi_1)$ is equivalent to the physical circuit obtained by permuting $\Code_1$ by $\pi_1^{-1}$, permuting $\Code_2$ by $\pi_2^{-1}$, applying $\CU(\pi)$, then permuting $\Code_1$ and $\Code_2$ back by $\pi_1$ and $\pi_2$ respectively.
    
    Using \Cref{cor:identity-automorphisms}, these physical permutations are logical identities, thus the logical action is the same as $\CU(\pi)$ itself.

\end{proof}

\begin{proposition}\label{prop:no-global-cnot}
    CNOT is not depth-one global CNOT parallel addressable on any two CSS codes with the same parameters, and asymptotical rates greater than $\frac{1}{3}$. 
\end{proposition}

\begin{proof}
     Let the two codes be $\overline{CSS}(A_1,B_1)$ and $\overline{CSS}(A_2,B_2)$.
     
     By studying the actions on stabilizers, we have that $\CNOT_{I,J}$ between two blocks of the code being valid implies that  
    \[ \forall \ a \in A_1, \ \pi(a \cap I) \in A_2 \]
    \[ \forall \ b \in B_2, \ \pi^{-1}(b \cap J) \in B_1 \]

    where $I$ is the set of control on the first block of the code, and $J$ the set of target on the second. $\pi$ is the bijective function going from control to target, and these equations show that it acts as an isomorphism from $A_1$ to $A_2$. Also $\pi^{-1}$ is an isomorphism from $B_2$ to $B_1$, so $\pi$ is an isomorphism from $B_1$ to $B_2$. Thus, each valid depth-one, global CNOT circuit corresponds to a valid isomorphism from the first code to the second, so our isomorphism upper bounds also apply to the number of such CNOT circuits.

    Moreover, \Cref{lem:cu-perm-invariant} tells us that these circuits have distinct actions only if the permutations are in distinct cosets of $\text{Sym}(H_{A_2})$, and the number of such distinct actions is bounded by $\frac{n!}{k_{max}!}$ by \Cref{thm:logical-permutation-count}.

    If we want these circuits to implement all parallel addressable logical CNOTs between the two codes, there must be $k!$ such gates. 
    
    The same counting arguments for the SWAP thus tell us that, asymptotically, there are not enough automorphisms to construct all such gates.

\end{proof}

\begin{proposition}\label{prop:no-global-CZ}
    $\CZ$ is not depth-one global $\CZ$ parallel addressable on any two CSS codes with the same parameters, and asymptotical rates greater than $\frac{1}{3}$. 
\end{proposition}
\begin{proof}
    The proof will proceed almost identically to \Cref{prop:no-global-cnot}. We first note that the required stabilizer relations are 
    \begin{align}
        \forall a \in A_1, \ &\pi(a) \in B_2\\
        \forall a\in A_2, \ &\pi^{-1}(a)\in B_1
    \end{align}
    (with the intersection with $I$ and $J$ not shown because the circuit is global).

    This tells us that $\pi(A_1)\subseteq B_2$ and $\pi(B_1)\subseteq A_2$. Since the codes have the same parameters, we know that $n-k=\dim(A_1)+\dim(B_1)=\dim(A_2)+\dim(B_2)$ (where $k$ is the number of logical qubits). Thus, $\pi(A_1)=B_2$ and $\pi(B_1)=A_2$, giving us an isomorphism $\pi$ from $A_1$ to $B_2$ and $B_1$ to $A_2$. Since \Cref{upp nbr permutation} does not care about the role of the codes, the counting arguments still apply, and still give the required upper bound.

\end{proof}

\begin{remark}
    The stabilizers relations stay valid in the context of qudits, which allows us to apply this result to the case of \cite{hvwz2025ccz}. To see this, in the context of qudits, the commutation rule of $\CZ_d$ gives : 
    \begin{align*}
    \CZ_d (X_d \otimes I) &= (X_d \otimes Z_d) \CZ_d \\
   \CZ_d (I \otimes X_d) &= (Z_d \otimes X_d) \CZ_d \\
    \CZ_d (Z_d \otimes I) &= (Z_d \otimes I) \CZ_d .\\
    \CZ_d (I \otimes Z_d) &= (I \otimes Z_d) \CZ_d .
    \end{align*}

Thus let $a \in A$ be a vector with coefficient in $\mathbb{F}_q$ representing a $X$ stabilizer for a CSS qudit code. Then a global depth-1 circuit of CZ sends it to a stabilizer iff $\pi(a) \in B$. The only difference here being that we use coefficient in $\mathbb{F}_q$, but the relation stays the same.
    
\end{remark}

\Cref{prop:no-global-CZ} relates to
Open Question 1 from \cite{hvwz2025ccz}: are there asymptotically good codes admitting transversal, addressable CCZ gates? The constructions they give for addressable CCZ gates are global, depth-one, and readily ``downgrade'' to addressable CZ gates. Thus, our results imply their techniques will not produce parallel addressable CCZ gates on an asymptotically good code with rates greater than $\frac{1}{3}$, unless one can (a) produce parallel addressable CCZ without simultaneously allowing parallel addressable CZ; or (b) produce non-global parallel addressable CCZ; or (c) use gates other than physical CZ or CCZ. Furthermore, using growing fields, \cite{hvwz2025ccz} provides codes with non-zero relative distance and rate that allow addressable CCZ using global circuits. In this case, they show that they can make both rate and relative distance at least $\frac{1}{6}$, but by reducing the relative distance, they can boost the rate to make it arbitrarily close to $\frac{1}{3}$. While their maximal rate matches our bound, it is important to note that we prove impossibility of parallel addressability, while they show the existence of addressability. Hence, it could be that they can achieve a better rate than $\frac{1}{3}$, but they would not be able to convert this method to parallel addressability.

Together, the results from \Cref{sec:swap-operators,sec:automorphism-cnots} suggest a concerning inability to create entanglement between the qubits in high-rate quantum code. We cannot apply CNOTs between arbitrary pairs of qubits in the high-rate code with permutations, and many families of CNOTs (say, all CNOTs from qubit $i$ to $i+1$) will generate all CNOTs, and thus are also impossible.

We might instead hope to use CNOTs to copy some data to another good quantum code, entangle the results there, and copy them back. However, \Cref{prop:no-global-cnot} tells us that an addressable CNOT between two good codes would only be able to send a given logical qubit in the first code to a small subset of logical qubits in the second code, and since the same no-go result would apply to the second, we would not be able to permute or entangle them before trying to copy them back. 

To escape these results, we note two critical assumptions: first, that \emph{both} codes have asymptotic rate at least $\frac{1}{3}$, and second, that the CNOTs are implemented by a permutation of the physical qubits. To escape the first, we might imagine copying to a less efficient code like the surface code. Such an architecture resembles caching, where quantum data is stored in the asymptotically good code, where computations are difficult, then copied into the surface code for computation. 

For the second restriction, one might imagine implementing logical CNOTs with physical CNOTs. However, notice that if the physical CNOTs also compose to physical swaps, then we run afoul of the same permutation counting arguments. This does not mean a CNOT cannot be implemented: our results do not forbid some CNOTs together with single-qubit Clifford gates to enact a logical CNOT.

\section{Algorithms for finding splits in CSS codes}\label{algoslpit}

We mentioned before that if a code splits then its distance is the minimum of the codes making it. Hence when we are trying to build the best code possible, we do not want them to split. In particular, the quantum Tanner code construction takes some code $C_A, C_B$ at random and shows that with some probability, it will give an asymptotically good code.  

Given a code, it is hard to compute its distance: it is equivalent to the problem of finding the smallest non-empty subset of dependent columns in the parity check matrix, and this is NP-hard. We thus expect that the splitting of the code gives a nice heuristic: if a code splits then most likely it will not have a good distance. Conversely, we hope that codes built this way with bad distance will split with good probability. This would give a better way to sample good quantum Tanner codes. 

In the following, we present two algorithms that detect if a code splits. The second approach detects and returns the splits in quadratic time (linear time for LDPC codes) in the number of qubits.

\subsection{System solving approach}

There is another equivalent way of defining splitting codes that we did not talk about in the splitting section.

\begin{proposition}\label{prop:diag-split-test}
    $\Code = \overline{\CSS}(A,B)$ is a splitting code if there exists a diagonal matrix $D$ with coefficients not all equal such that for all $a\in A$ and $b\in B$, $aD\in A$ and $bD\in B$.
\end{proposition}

\begin{proof}

    Starting by the easy direction, assuming that $C$ splits on some support $h$, then both $A,B$ split on $h$. Let us take $D$ defined as $D_{i,i} = 1$ if $i \in h$ and $0$ otherwise, since $C$ splits we have that the coefficients of $D$ are not all equal. Multiplication by $D$ is just a projection onto $h$, so $aD\in A$ and $bD\in B$ for all $a\in A$ and $b\in B$, as required.

    In the other direction, let us first show that for any polynomial with integer coefficients $P$, we have $aP(D) \in A$ and $bP(D)\in B$ for any $a,b\in A\times B$. By assumption we have $aD\in A$, hence applying this rule again gives that for all $ s \in \mathbb{N}, aD^s\in A$. Finally, since $A$ is linear, we get that for all polynomials $P$ with integer coefficients, $aP(D)\in A$. The same reasoning applies to $B$. 
    
    Now, we also know that the coefficients of $D$ are not all equal, hence there exists two non-empty complementary sets $U,V$ of indices such that the values of $D$ at indices in $U$ and $V$ are always different and $U\cup V=\llbracket n\rrbracket$. Consider $P_U$ the interpolation polynomial such that $P_U(D_{i,i}) = 1$ if $i \in U$ and $0$ otherwise. Let $P_V$ be defined the same way over $V$. We thus have $(P_U + P_V)(D) = I_{n}$. Thus $A P_U(D) + A P_V(D) = A I_{n} = A$, and $A P_U(A) \subseteq A$, $A P_V(A) \subseteq A$. Calling $A_U = A P_U(D)$ and $A P_V(D)$ we see that $A$ splits into $A_U, A_V$ on support $h = U$, since $P_U(D)$ and $P_V(D)$ are orthogonal projections. 

    The same procedure splits $B$ into $B_U:=BP_U(D)$ and $B_V:=BP_V(D)$, which have the same supports $h$ and $\llbracket n\rrbracket\setminus h$, so $\Code$ splits as well.

\end{proof}

\begin{remark}
    In our case, since we work with binary vectors, this is exactly the same definition as the one with the support $h$ on which we project. But it still works in the non binary case, and might be more interesting. However, this definition is great as it gives a natural idea for an algorithm checking if a code splits. 
\end{remark}

\paragraph{}
For a code $\overline{\CSS}(A,B)$ with generators and parity checks $G_A,H_A,G_B,H_B$, let us consider the following equation : 
\begin{equation}
    \begin{pmatrix}
        G_A & G_B
    \end{pmatrix}\begin{pmatrix}
        D & 0 \\ 0 & D 
    \end{pmatrix}\begin{pmatrix}
        H_A \\
        H_B
    \end{pmatrix}=0
\end{equation}
 when $D$ is a diagonal matrix on $n$ qubits. We know that $D = \lambda I_n$ is always a solution of the equation as it preserves the codespace. Furthermore, if $D$ is diagonal with coefficients not all equal, then it is in the solution space if and only if the code splits by \Cref{prop:diag-split-test}. This gives rise to \Cref{alg:split-testing}.

\begin{algorithm}[H]
\caption{Split testing}\label{alg:split-testing}
\begin{algorithmic}[1]
\REQUIRE Generator and parity checks $G_A,H_A$ and $G_B,H_B$ of the codes $A,B$
\ENSURE Detect if the code splits
\STATE $\mathcal{S} \leftarrow \text{Solve}_D \footnotesize \left( \begin{pmatrix}
        G_A^T \\ G_B^T
    \end{pmatrix}^T\begin{pmatrix}
        D & 0 \\ 0 & D 
    \end{pmatrix}\begin{pmatrix}
        H_A \\
        H_B
    \end{pmatrix}=0\right)$. 
\STATE $d \leftarrow \text{dim}(\mathcal{S})$
\RETURN $d > 1 $ ? 
\end{algorithmic}
\end{algorithm}

\paragraph{}
This algorithm has complexity $\mathcal{O}(n^\omega)$ which represents the time complexity of solving a system on $\mathcal{O}(n)$ qubits ($2 \leq \omega < 3$). This algorithm is intuitive after seeing this new definition but it is not optimal. The graph theoretical approach will give the splits explicitly in quadratic time.

\subsection{Graph theoretical approach}

In \cite{BurnistonJohn2023} the author develops a method to identify if a matrix is "reducible" where $F$ being reducible means that it either has a column of zeros or that there is an invertible map $U$ and a permutation $P$ such that $UFP= \begin{pmatrix}
    F_1 & 0 \\ 
    0 & F_2
\end{pmatrix}$.

In our case, the definition of reducibility is the same as the definition of splitting of classical codes. Now to adapt in our case, we need to obtain the splits of $X,Z$ stabilizers but also compare them and find a common split. If such a split exists then the code splits.

The following algorithm follows the method from \cite{BurnistonJohn2023} to find if a matrix splits. Intuitively, each qubit represent a vertex, and we draw an edge between two qubits if there exists a stabilizer in which they both appear in the row reduced version of the $X$ and $Z$ stabilizer generators. To make it easier and more efficient we consider another graph which is the Tanner graph: instead of having each edge labeled by a stabilizer, qubits have an edge to a stabilizer if they appear in them. We then get the connected components of this graph. Each connected component represents a part of the split of the code as well as the stabilizer it concerns. We can then extract the qubits from those components if we only care about the splits. The proof of correctness of this approach can be derived from the one in \cite{BurnistonJohn2023}.

Given generators of its stabilizer group represented as a $r \times n$ matrix, this algorithm returns in time $\mathcal{O}(n + c \times n)$ the split decomposition of a CSS code, where $c$ is the maximal number of qubits in a stabilizer of the generator. Hence for LDPC codes, this algorithm is linear, while it is quadratic in general. 

\begin{algorithm}[H]
\caption{Common Block Diagonalization}
\begin{algorithmic}[1]
\REQUIRE Stabilizer matrix $\mathcal{S}= \begin{bmatrix} \mathcal{S}_X & 0 \\ 0 & \mathcal{S}_Z \end{bmatrix}$
\ENSURE Partition of columns for common block diagonalization

\STATE $\mathcal{S}_X' \leftarrow \text{Row\_Reduced\_Form}(\mathcal{S}_X)$
\STATE $\mathcal{S}_Z' \leftarrow \text{Row\_Reduced\_Form}(\mathcal{S}_Z)$
\STATE $\mathcal{S}' \leftarrow \begin{bmatrix}\mathcal{S}_X' & 0 \\ 0 & \mathcal{S}_Z'\end{bmatrix}$
\STATE $G \leftarrow \text{Tanner\_Graph}(\mathcal{S}')$
\STATE $C \leftarrow \text{Connected\_components}(G)$
\STATE $\text{Blocks} \leftarrow \text{Blocks\_From\_Connected\_Components}(C)$
\RETURN $\text{Blocks}$

\end{algorithmic}
\end{algorithm}

\begin{algorithm}[H]
\caption{Tanner\_Graph}
\begin{algorithmic}[1]
\REQUIRE Matrix $F$ ( in symplectic form ) 
\ENSURE Bipartite graph $G$ with edges between $r_i$ and $c_j$ if $F_{i,j} = 1$ or $F_{i,j+n} = 1$.

\STATE Initialize graph $G = (R, C, E)$ where $R$ are row vertices, $C$ are column vertices, and $E$ are edges
\FOR{each entry $F_{ij}$ in $F$}
    \IF{$F_{ij} = 1 \OR F_{i,j+n} = 1$}
        \STATE Add edge $(r_i, c_j)$ to $E$
    \ENDIF
\ENDFOR

\RETURN $G$

\end{algorithmic}
\end{algorithm}

\begin{algorithm}[H]
\caption{Blocks\_From\_Connected\_Components}
\begin{algorithmic}[1]
\REQUIRE Connected components $C$
\ENSURE Column index for the block diagonalization of the matrix 

\STATE $\text{Blocks} \leftarrow \emptyset$

\FOR{$C_i$ in $C$}
    \STATE $\text{block}_i \leftarrow \emptyset$
    \FOR{$c_j$ column vertex in $C_i$}
        \STATE $\text{block}_i \leftarrow \text{block}_i \cup \{ j \}$
    \ENDFOR
    \STATE $\text{Blocks} \leftarrow \text{Blocks} \cup \{ \text{block}_i \}$
\ENDFOR 
\RETURN $\text{Blocks}$

\end{algorithmic}
\end{algorithm}

\section*{Acknowledgements}
We thank Adam Wills and Rachel Zhang for helpful discussions on the connections between our results. We also thank Arthur Pesah for introducing us to the addressability problem and for helpful discussions. 

S. Jaques acknowledges the support of the Natural Sciences and Engineering Research Council of Canada (NSERC), funding reference number RGPIN-2024-03996.

\bibliographystyle{quantum}
\bibliography{citations}

\section*{Appendix}\label{appendix}

\begin{lemma}\label{integral_pas_beau}
    For all sequences $(\rho_n)$ and $(\rho'_n)$ such that $ \rho_n + \rho'_n > 1$, 
    \[ \exists n_0 \in \mathbb{N} \text{ such that } \forall \ n \geq n_0, \ \frac{n!}{(\rho'_n n)!} < (\rho_n n)! \]
\end{lemma}
\begin{proof} 
    Since we are in the positive part of the logarithm, and it is an increasing function, it will be equivalent to prove that
    $\sum_{\rho'_n n \leq i \leq n} \log(i) < \sum_{2 \leq j \leq \rho_n n} \log(j)$.
    
    Again because the logarithm is increasing, we can make integral inequalities:
    \[ \sum_{\rho'_n n \leq i \leq n} \log(i) \leq \int_{\rho'_n n}^{n+1} \log(x) dx \] 
    \[ \sum_{2 \leq i \leq \rho_n n} \log(i) \geq \int_{1}^{\rho_n n-1} \log(x) dx \]

Now we use $\int_{\rho'_n n}^{n+1} \log(x) dx = g(n+1)-g(\rho'_n n)$ and $\int_{1}^{\rho_n n-1} \log(x) dx = g(\rho_n n-1) - g(1)$, where $g(x) = x\log(x)-x$. . 
    Hence, if $g(n+1)-g(\rho'_n n) < g(\rho_n n-1) -1$ then the inequality holds. We can do some algebra to see that this holds when the following expression is negative:
    $g(n+1) - g(\rho'_n n) - g(\rho_n n-1) + g(1)  = (1-\rho'_n - \rho_n)n\log(n) + \mathcal{O}(n)$.

    Since we assumed $\rho'_n + \rho_n > 1$, it gives that there exists an $n_0$ for which $\forall \ n\geq n_0$, $\frac{n!}{(\rho' n)!} < (\rho n)! $

\end{proof}

\begin{corollary}\label{cor:count-double-factorial}
    For all sequences $(\rho_n)$ and $(\rho'_n)$ such that $\frac{1}{2}\rho_n+\rho'_n>1$,
    \[ \exists n_0 \in \mathbb{N} \text{ such that } \forall \ n \geq n_0, \ \frac{n!}{(\rho'_n n)!} < (\rho_n n)!! \]
\end{corollary}
\begin{proof}
    Recall that the double factorial is $x!!=x(x-2)(x-4)\dots$, and thus for odd $x$ its logarithm can instead be written as $\sum_{1 < j < \frac{1}{2}\rho_n n}\log(2j+1)$, and thus lower-bounded by 
    \[\int_1^{\frac{1}{2}\rho_n n - 1}\log(2x+1)dx =\frac{1}{2}\left(g(\rho_n n - 1) - g(3)\right).\] 
    and the final asymptotic expression is $(1-\rho'_n-\frac{1}{2}\rho_n)n\log(n)+\mathcal{O}(n)$, giving the result.

\end{proof}

\begin{lemma}\label{lem:count-GLn}
    For any integer $q>1$ and all sequences $(\rho_n)$, $(\rho'_n)$ such that $\rho_n,\rho'_n>0$ and $\rho_n > \sqrt{\frac{\log n}{n\log q}} + \Omega\left(\frac{1}{\sqrt{n}}\right)$, 
    \[
         \exists n_0 \in \mathbb{N} \text{ such that } \forall \ n \geq n_0, \ \frac{n!}{(\rho'_n n)!} < q^{(\rho_n n)^2-1} 
    \]
    In particular, $\frac{n!}{(\rho'_n n)!} < \vert \text{GL}_{\rho_n n}(q) \vert$.
\end{lemma}
\begin{proof}
    We apply precisely the same reasoning as the last two lemmas, and obtain that the result holds when the following expression is negative:
    \begin{equation*}
        g(n+1) - g(\rho'_n n) - \log(q)(\rho_n n)^2 + 1
    \end{equation*}
    which works out to 
    \begin{align*}
        &(n+1)\log(n+1) - n -(\rho'_n n)\log(\rho'_n n) - \log(q)\rho_n^2n^2 \\
        =&  (1-\rho'_n)\frac{\log n}{n} - \log(q)\rho^2_n n^2 O\left(\frac{\log n}{n^2}\right).
    \end{align*}
    The lower bound on $\rho_n$ shows that this will be negative.

    To show that this relates to $\text{GL}_{\rho_n n}(q)$, we note that 
    \begin{equation}
        \vert \text{GL}_k(q)\vert = \prod_{i=0}^{k-1} (q^k - q^i)
    \end{equation}
    which can be lower-bounded by $q^{k^2-1}$.

\end{proof}

\end{document}